\documentclass[sigconf]{acmart}

\usepackage{graphicx}
\usepackage{epsfig}
\usepackage{url}
\usepackage{hyperref}
\usepackage{caption}
\usepackage[ruled, vlined, linesnumbered]{algorithm2e}
\usepackage[show]{chato-notes}
\usepackage{xspace}
\usepackage{multirow} 

\hypersetup{
    bookmarks=true,         
    unicode=false,          
    pdftoolbar=true,        
    pdfmenubar=true,        
    pdffitwindow=false,     
    pdfstartview={FitH},    
    pdftitle={My title},    
    pdfauthor={Author},     
    pdfsubject={Subject},   
    pdfcreator={Creator},   
    pdfproducer={Producer}, 
    pdfnewwindow=true,      
    colorlinks=true,       
    linkcolor=black,          
    citecolor=black,        
    filecolor=black,      
    urlcolor=black           
}

\newtheorem{mydefinition}{Definition}
\newtheorem{myproposition}{Proposition}
\newtheorem{mytheorem}{Theorem}
\newtheorem{mycorollary}{Corollary}
\newtheorem{myexample}{Example}
\newtheorem{mylemma}{Lemma}
\newtheorem{problem}{Problem}

\newtheorem{myobservation}{Observation}
\newcommand{\core}{C_{k, \Delta}}

\newcommand{\baseline}{\textsf{Na\"ive-span-cores}}
\newcommand{\cores}{\textsf{Span}-\textsf{cores}}
\newcommand{\innermosts}{\textsf{Maximal-span-cores}}
\newcommand{\baselineinnermosts}{\textsf{Na\"ive-maximal-span-cores}}
\newcommand{\innermost}{\textsf{innermost-core}}

\newcommand{\spancore}{span-core\xspace}
\newcommand{\Spancore}{Span-core\xspace}
\newcommand{\spancores}{\spancore{s}\xspace}

\newcommand{\maxspancoreprob}{\textsc{Maximal \Spancore Mining}\xspace}

\newcommand{\coresset}{\mathbf{C}}
\newcommand{\coressetdelta}{\mathbf{C}_\Delta}

\newcommand{\imcores}{\mathbf{C}_M}

\newcommand{\bigO}{\mathcal{O}}
\newcommand{\tdeg}{\mbox{\ensuremath{d}}}

\newcommand{\spara}[1]{\smallskip\noindent{\bf #1}}

\newcommand{\squishlist}{
 \begin{list}{$\bullet$}
  {  \setlength{\itemsep}{0pt}
     \setlength{\parsep}{3pt}
     \setlength{\topsep}{3pt}
     \setlength{\partopsep}{0pt}
     \setlength{\leftmargin}{2em}
     \setlength{\labelwidth}{1.5em}
     \setlength{\labelsep}{0.5em}
} }
\newcommand{\squishlisttight}{
 \begin{list}{$\bullet$}
  { \setlength{\itemsep}{0pt}
    \setlength{\parsep}{0pt}
    \setlength{\topsep}{0pt}
    \setlength{\partopsep}{0pt}
    \setlength{\leftmargin}{2em}
    \setlength{\labelwidth}{1.5em}
    \setlength{\labelsep}{0.5em}
} }

\newcommand{\squishdesc}{
 \begin{list}{}
  {  \setlength{\itemsep}{0pt}
     \setlength{\parsep}{3pt}
     \setlength{\topsep}{3pt}
     \setlength{\partopsep}{0pt}
     \setlength{\leftmargin}{1em}
     \setlength{\labelwidth}{1.5em}
     \setlength{\labelsep}{0.5em}
} }

\newcommand{\squishend}{
  \end{list}
}

\copyrightyear{2018} 
\acmYear{2018} 
\setcopyright{acmlicensed}
\acmConference[CIKM '18]{The 27th ACM International Conference on Information and Knowledge Management}{October 22--26, 2018}{Torino, Italy}
\acmBooktitle{The 27th ACM International Conference on Information and Knowledge Management (CIKM '18), October 22--26, 2018, Torino, Italy}
\acmPrice{15.00}
\acmDOI{10.1145/3269206.3271767}
\acmISBN{978-1-4503-6014-2/18/10}

\begin{document}

\title[Mining (maximal) span-cores from temporal networks]{Mining (maximal) span-cores from temporal networks}

\author{Edoardo Galimberti}
\affiliation{\small
  \institution{ISI Foundation, Italy}
  \institution{University of Turin, Italy}
}
\email{edoardo.galimberti@isi.it}

\author{Alain Barrat}
\affiliation{\small
  \institution{Aix Marseille Univ, CNRS, CPT,  France}
\institution{ISI Foundation, Italy}
}
\email{alain.barrat@cpt.univ-mrs.fr}

\author{Francesco Bonchi}
\affiliation{\small
  \institution{ISI Foundation, Italy}
  \institution{Eurecat, Barcelona, Spain}
}
\email{francesco.bonchi@isi.it}

\author{Ciro Cattuto}
\affiliation{\small
  \institution{ISI Foundation, Italy}
}
\email{ciro.cattuto@isi.it}

\author{Francesco Gullo}
\affiliation{\small
  \institution{UniCredit, R\&D Dept., Italy}
}
\email{gullof@acm.org}

\renewcommand{\shortauthors}{E. Galimberti~\emph{et al.}}

\begin{abstract}
When analyzing temporal networks, a fundamental task is the identification of dense structures (i.e., groups of vertices that exhibit a large number of links), together with their temporal span (i.e., the period of time for which the high density holds). We tackle this task by introducing a notion of temporal core decomposition where each core is associated with its span: we call such cores \mbox{\emph{span-cores}}.

As the total number of time intervals is quadratic in the size of the temporal domain $T$ under analysis,
the total number of span-cores is quadratic in $|T|$ as well.
Our first contribution is an algorithm that, by exploiting containment properties among span-cores, computes all the span-cores efficiently.
Then, we focus on the problem of finding only the \emph{maximal span-cores}, i.e., span-cores that are not dominated by any other span-core by both the coreness property and the span.
We devise a very efficient algorithm that exploits theoretical findings on the maximality condition to directly compute the maximal ones
without computing all span-cores.

Experimentation on several real-world temporal networks confirms the efficiency and scalability of our methods.
Applications on temporal networks, gathered by a proximity-sensing infrastructure recording face-to-face interactions in schools, highlight the relevance of the notion of (maximal) span-core in analyzing social dynamics and detecting/correcting anomalies in the data.
\end{abstract}

\maketitle \sloppy

\section{Introduction}
\label{sec:introduction}


\enlargethispage{\baselineskip}
A temporal network is a representation of entities (vertices), their relations (links), and how these relations are established/broken along time. Extracting dense structures (i.e., groups of vertices exhibiting a large number of links among each other), together with their temporal span (i.e., the period of time for which the high density is observed) is a key mining primitive. This type of patterns enables fine-grain analysis of the network dynamics
and can be a building block towards more complex tasks (such as finding temporally recurring subgraphs or anomalously dense ones) and applications. 
For instance, they can help in studying the contact networks among individuals
to quantify the transmission opportunities of respiratory infections, modeling situations where the risk of transmission is higher,
with the goal of designing mitigation strategies~\cite{Gemmetto2014}. Anomalously dense temporal patterns among entities in a co-occurrence graph (e.g., extracted from the Twitter stream) have also been used to identify, in real-time, events and buzzing stories  \cite{Angel,BonchiBGS16}. In scientific collaboration
and citation networks these patterns can help understand the dynamics of collaboration in successful
professional teams, study the evolution of scientific topics, and detect emerging technologies \cite{erdi2013prediction}.

In this paper we adopt as measure of density of a pattern the \emph{minimum degree} holding among the vertices in the subgraph during the pattern's span. 
The problem of extracting \emph{all} these patterns is tackled by introducing a notion of \emph{temporal core decomposition} in which each core is associated with its \emph{span}, i.e., an interval of
\emph{contiguous timestamps}, for which the coreness property holds.

To the best of our knowledge, this type of core, which we call \emph{span-core}, has never been studied so far.

\enlargethispage{\baselineskip}
\spara{Challenges and contributions.}
As the total number of time intervals is quadratic in the size of the temporal domain $T$ under analysis, also the total number of span-cores is, in the worst case, quadratic in $T$. Nevertheless, exploiting nice containment properties we devise an efficient algorithm for computing all  the  \emph{span-cores}.
Then, we shift our attention to the problem of finding only the \emph{maximal span-cores}, i.e., span-cores that are not dominated by any other span-core by both the coreness property and the span. A straightforward way of approaching the maximal-\spancore-mining problem is to filter out non-maximal \spancores during the execution of an algorithm for computing the whole \spancore decomposition.
However, as the maximal ones are usually much less than the overall span-cores, it would be desirable to have a method that effectively exploits the maximality property and extracts maximal \spancores directly, without computing a complete decomposition.
The design of an algorithm of this kind is an interesting challenge, as it contrasts the intrinsic conceptual properties of core decomposition, based on which a core of order $k$ can be efficiently computed from the core of order $k\!-\!1$, of which it is a subset. For this reason, at first glance, the computation of the core of the highest order would seem as hard as computing the overall core decomposition. Instead, in this work we derive a number of theoretical properties about the relationship among \spancores of different temporal intervals and, based on these findings, we show how such a challenging goal may be achieved.

The contributions of this paper can be summarized as follows:
\squishlist
\item
We introduce the notion of span-core decomposition and maximal span-core in temporal networks. We characterize structure and size of the search space, and prove important containment properties  (Section~\ref{sec:problem}).

\item
We devise an algorithm for computing all span-cores that exploits the aforementioned containment properties and is orders of magnitude faster than a na\"{\i}ve method based on traditional core decomposition (Section~\ref{sec:algorithms}).

\item
We study the problem of finding only the maximal span-cores.
We derive a number of theoretical findings about the relationship among maximal span-cores and exploit them to devise an algorithm that is more efficient than computing all span-cores and discarding the non-maximal ones (Section~\ref{sec:algorithms:maximal}).

\item
We provide a comprehensive experimentation on several real-world temporal networks, with millions of vertices, tens of millions of edges, and hundreds of  timestamps, which attests efficiency and scalability of our methods (Section~\ref{sec:experiments}).

\item
We present applications on face-to-face interaction networks, that illustrate the relevance of the notion of (maximal) span-core in real-life analyses (Section~\ref{sec:casestudies}).

\squishend

The next section overviews the related literature, while Section~\ref{sec:conclusions} discusses future work and concludes the paper.
\section{Background and related work}
\label{sec:related}

\enlargethispage{\baselineskip}
\textbf{Core decomposition.} 
In standard graphs, among the many definitions of dense structures, \emph{core decomposition} plays a central role as it can be computed in linear time~\cite{MatulaB83,batagelj2011fast},
and can speed-up/approximate dense-subgraph extraction according  to various other definitions.
%
For instance, core decomposition allows for finding cliques more efficiently~\cite{EppsteinLS10},
it can be used to approximate the densest-subgraph problem \cite{KortsarzP94}, and betweenness centrality~\cite{HealyJMA06}.

Given a simple graph $G=(V,E)$, let $d(S,u)$ denote the degree of vertex $u \in V$ in the subgraph induced by vertex set $S \subseteq V$, i.e., $d(S,u) = |\{v \in S \mid (u,v) \in E \}|$.

\begin{mydefinition}[Core Decomposition]\label{def:kcores}
 The $k$\emph{-core} (or core of order $k$) of $G$ is a
 \emph{maximal} set of vertices $C_k \subseteq V$ such that $\forall u \in C_k: d(C_k,u) \geq k$.
 The set of all $k$-cores $V = C_0 \supseteq C_1 \supseteq \cdots \supseteq C_{k^*}$ ($k^* = \arg\max_{k} C_k \neq \emptyset$) is the
 \emph{core decomposition} of $G$.
\end{mydefinition}

Core decomposition has been established as an important tool to analyze and visualize complex networks~\cite{DBLP:conf/gd/BatageljMZ99,Alvarez-HamelinDBV05}
in several domains, e.g., bioinformatics~\cite{DBLP:journals/bmcbi/BaderH03,citeulike:298147},
software engineering~\cite{DBLP:journals/tjs/ZhangZCLZ10},
and social networks~\cite{Kitsak2010,GArcia2013}. 
It has been studied under various settings, such as distributed~\cite{DistributedCores1}, streaming/maintenance~\cite{StreamingCores,li2014efficient}, and disk-based~\cite{DiskCores}, and for various types of graph, such as uncertain~\cite{bonchi14cores}, directed~\cite{DirectedCores},
and weighted graphs~\cite{WeigthedCores}.

Core decomposition in \emph{multilayer networks} has been studied in~\cite{GalimbertiBG17}. As any subset of layers is allowed in this setting, the total number of cores is intrinsically exponential. Although temporal networks can be seen as a special case of multilayer networks (where each timestamp is interpreted as a layer), the sequentiality of time represents an important structural constraint: in this paper we are interested in cores that span a temporal interval, and not simply any subset of (potentially non-contiguous) timestamps.
As a consequence, the search space and the number of cores are no longer exponential as in the multilayer case. 
A type of core decomposition for temporal networks has been proposed by Wu~{\em et~al.}~\cite{wu2015core}, who define the $(k,h)$-core  as the largest subgraph in which every vertex has at least $k$ neighbors and at least $h$ temporal connections with each of them.
Therefore, even in the Wu~{\em et~al.}'s definition the sequentiality of connections is not taken into account and non-contiguous timestamps can support the same core.
Our temporal cores have instead a clear temporal collocation and continuous spans, thus the Wu~{\em et~al.}'s definition cannot be reduced to ours (or vice versa).
As we will see in Section~\ref{sec:casestudies}, such a temporal collocation is important in applications.

\spara{Patterns in temporal networks.}
Semertzidis~\emph{et~al.}~\cite{semertzidis2016best} introduce the problem of identifying a set of vertices that are densely connected in all or at least $k$ timestamps of a temporal network.
Similarly, Jethava~and~Beerenwinkel~\cite{jethava2015finding} formulate the densest-common-subgraph problem on an input that can be interpreted as a special type of temporal network, i.e., a set of graphs sharing the same vertex set.
The notion of $\Delta$-clique has been proposed in~\cite{viard2016computing, himmel2016enumerating}, as a set of vertices in which each pair is in contact at least every $\Delta$ timestamps.
Complementary approaches study the problem of discovering dense temporal subgraphs whose edges occur in short time intervals considering the exact timestamp of the occurrences~\cite{rozenshtein2017finding}, and the problem of maintaining the densest subgraph in the dynamic graph model~\cite{epasto2015efficient}.
A slightly different, but still related body of literature focuses on frequent evolution patterns in temporal attributed graphs~\cite{BerlingerioBBG09,inokuchi2010mining, desmier2012cohesive},
link-formation rules in temporal networks~\cite{BringmannBBG10,leung2010mining}, and
the discovery of dynamic relationships and events~\cite{das2011dynamic}
or of correlated activity patterns~\cite{Gauvin2014}.

\section{Problem Definition}
\label{sec:problem}

\enlargethispage{\baselineskip}
We are given a \emph{temporal graph} $G = (V,T,\tau)$, where $V$ is a set of vertices,  $T = [0, 1, \ldots, t_{max}] \subseteq \mathbb{N}$ is a discrete time domain, and $\tau: V  \times V \times  T\rightarrow \{0,1\}$ is a function defining for each pair of vertices $u,v \in V$ and each timestamp $t \in T$ whether edge $(u,v)$ exists in $t$. We denote $E = \{(u,v,t) \mid \tau(u,v,t) = 1 \}$ the set of all temporal edges. Given a timestamp $t \in T$, $E_t = \{(u,v) \mid \tau(u,v,t) = 1 \}$ is the set of edges existing at time $t$.
A temporal interval $\Delta = [t_s, t_e]$ is contained into another temporal interval $\Delta' = [t'_s, t'_e]$, denoted $\Delta \sqsubseteq \Delta'$, if $t'_s \leq t_s$ and $t'_e \geq t_e$.
Given an interval $\Delta \sqsubseteq T$, we denote $E_\Delta = \bigcap_{t \in \Delta} E_t$ the edges existing in \emph{all timestamps} of $\Delta$. Given a subset $S \subseteq V$ of vertices, let $E_{\Delta}[S] = \{(u,v) \in E_{\Delta} \mid u \in S, v \in S\}$ and $G_{\Delta}[S] = (S, E_{\Delta}[S])$.
Finally, the  \emph{temporal degree} of a vertex $u$ within $G_{\Delta}[S]$ is denoted $\tdeg_\Delta(S,u) = |\{v \in S \mid (u,v) \in E_\Delta[S] \}|$.

\begin{mydefinition}[$(k,\Delta)$-core] \label{def:core}
The  $(k,\Delta)$\emph{-core} of a temporal graph $G = (V,T,\tau)$ is (when it exists) a maximal and non-empty set of vertices $ \emptyset \neq C_{k,\Delta} \subseteq V$, such that $\forall u \in C_{k,\Delta} : \tdeg_\Delta(C_{k,\Delta},u) \geq k$, where
 $\Delta \sqsubseteq T$ is a temporal interval and $k \in \mathbb{N}^+$.
\end{mydefinition}

A $(k,\Delta)$-core is a set of vertices implicitly defining a cohesive subgraph (where $k$ represents the cohesiveness constraint), together with its \emph{temporal span}, i.e., the interval $\Delta$ for which the subgraph satisfies the cohesiveness constraint.
In the remainder of the paper we refer to this type of temporal pattern as \emph{\spancore}.

The first problem we tackle in this work is to compute the \emph{span-core decomposition} of a temporal graph $G$, i.e., all span-cores of~$G$.

\begin{problem}[Span-core decomposition] \label{pbl:dececomposition}
Given a temporal graph $G$, find the set of all $(k,\Delta)$-cores of $G$.
\end{problem}

Unlike standard cores of simple graphs, \spancores are not all nested into each other, due to their spans.
However, they still exhibit containment properties.
Indeed, it can be observed that a $(k,\Delta)$-core is contained into any other $(k',\Delta')$-core with less restrictive degree  and span conditions, i.e., $k' \leq k$, and $\Delta' \sqsubseteq \Delta$.
This property is depicted in Figure~\ref{fig:searchspace}, and formally stated in the next proposition.

\begin{myproposition}[\Spancore containment]\label{prp:conteinment}
For any two \spancores $C_{k, \Delta}$, $C_{k',\Delta'}$ of a temporal graph $G$ it holds that
$$
k' \leq k \wedge \Delta' \sqsubseteq \Delta \ \Rightarrow \  \core \subseteq C_{k',\Delta'}.
$$
\end{myproposition}
\begin{proof}
The result can be proved by separating the two conditions in the hypothesis, i.e., by separately showing that ($i$) $k' \leq k \Rightarrow C_{k, \Delta} \subseteq C_{k', \Delta}$, and ($ii$) $\Delta' \sqsubseteq \Delta \Rightarrow C_{k,\Delta} \subseteq C_{k,\Delta'}$.
The first argument holds as, keeping the span $\Delta$ fixed, the maximal set of vertices $C$ for which $\tdeg_\Delta(C,u) \geq k$ is clearly contained in the maximal  set of vertices $C'$ for which $\tdeg_\Delta(C',u) \geq k'$, if $k' \leq k$. As far as the second argument, it can be noted that $\Delta' \sqsubseteq \Delta \Rightarrow E_\Delta \subseteq E_{\Delta'}$, which implies that
$\forall u \in C_{k,\Delta} : \tdeg_\Delta(C_{k,\Delta}, u) \leq \tdeg_{\Delta'}(C_{k,\Delta}, u)$.
Therefore, all vertices within $C_{k,\Delta}$ satisfy the condition to be part of $C_{k,\Delta'}$ too.
\end{proof}

\begin{myobservation}\label{observation1}
For a fixed temporal interval $\Delta \sqsubseteq T$, finding all span-cores that have $\Delta$ as their span is equivalent to computing the classic core decomposition~\cite{batagelj2011fast} of the simple graph $G_\Delta = (V, E_\Delta)$.
\end{myobservation}

As the total number of temporal intervals that are contained into the whole time domain $T$ is $|T|(|T|\!+\!1)/2$, the total number of \spancores  is $\mathcal{O}(|T|^2 \times k_{max})$, where $k_{max}$ is the largest value of $k$ for which a  $(k,\Delta)$-core exists.
The number of \spancores is thus quadratic in $|T|$, which may be too large an output for human inspection.
In this regard, it may be useful to focus only on the most relevant cores, i.e., the \emph{maximal} ones, as defined next.

\begin{mydefinition}[Maximal \Spancore] \label{def:maximal}
A \spancore $\core$ of a temporal graph $G$ is said \emph{maximal} if there does not exist any other \spancore $C_{k',\Delta'}$ of $G$ such that $k \leq k'$ and $\Delta \sqsubseteq \Delta'$.
\end{mydefinition}

Hence, a \spancore is recognized as maximal if it is not dominated by another \spancore both on the order $k$ and the span $\Delta$.
Differently from the \emph{innermost core} (i.e., the core of the highest order) in the classic core decomposition, which is unique,
in our temporal setting the number of maximal \spancores is $\mathcal{O}(|T|^2)$, as, in the worst case,
there may be one maximal \spancore for every temporal interval.
However, as observed experimentally, maximal \spancores are always much less than the overall \spancores: the difference is usually one order of magnitude or more.
The second problem we tackle in this work is to compute the maximal \spancores of a temporal graph.

\begin{problem}[\maxspancoreprob] \label{pbl:maximal}
Given a temporal graph $G$, find the set of all maximal $(k,\Delta)$-cores of $G$.
\end{problem}

Clearly, one could solve Problem~\ref{pbl:maximal} by  solving Problem~\ref{pbl:dececomposition} and  filtering out all the non-maximal \spancores.
However, an interesting yet challenging question (Section~\ref{sec:algorithms:maximal}) is whether one can exploit the maximality condition to develop faster algorithms that can directly extract the maximal ones, without computing all the span-cores.

\begin{figure}[t!]
\vspace{-3mm}
\centering
\includegraphics[width=0.8\columnwidth]{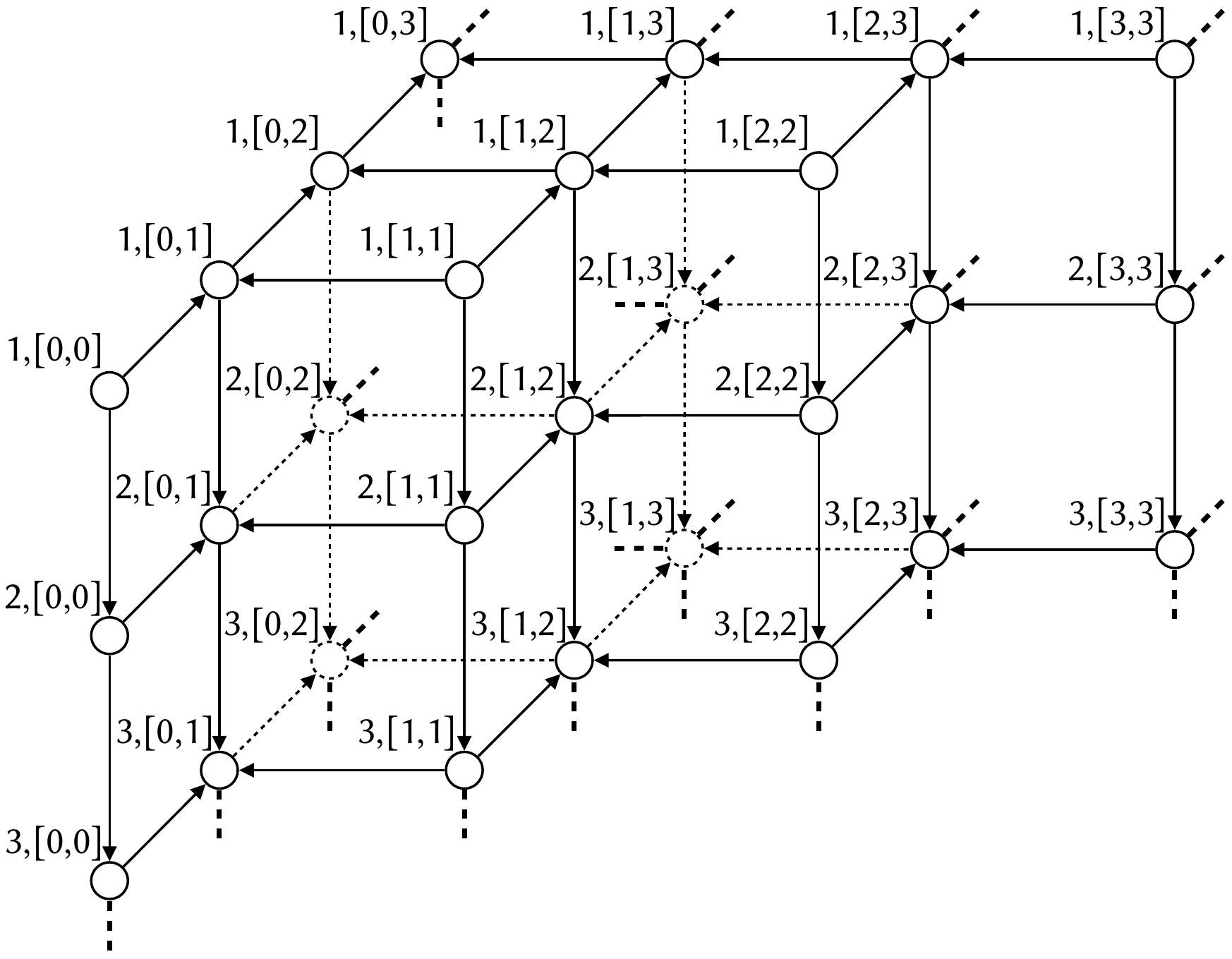}
\vspace{-2mm}
\caption{\label{fig:searchspace} {Search space: for a temporal span $\Delta = [t_s,t_e]$, the $(k,\Delta)$-core is depicted as a node labeled ``$k, [t_s, t_e]$''. An arrow  $C_1 \rightarrow C_2$ denotes $C_1 \supseteq C_2$} (distinction between solid and dotted arrows is for visualization sake only).}
\vspace{1mm}
\end{figure}

\section{Computing all Span-cores}
\label{sec:algorithms}


In this section we devise algorithms for computing a complete \spancore decomposition of a temporal graph (Problem~\ref{pbl:dececomposition}).

\spara{A na\"ive approach.} As stated in Observation 1, for a fixed temporal interval $\Delta \sqsubseteq T$, mining all span-cores $\core$ is equivalent to computing the classic core decomposition of the graph $G_\Delta = (V, E_\Delta)$.
A na\"{\i}ve strategy is thus to run a core-decomposition subroutine~\cite{batagelj2011fast}
on graph $G_\Delta$ for each temporal interval $\Delta \sqsubseteq T$. Such a method has time complexity $\bigO(\sum_{\Delta \sqsubseteq T} (|\Delta| \times |E|))$, i.e., $\bigO(|T|^2 \times |E|)$.

\spara{A more efficient algorithm.} Looking at Figure \ref{fig:searchspace} one can observe that the  na\"ive algorithm only exploits one dimension of the containment property: it starts from each point on the top level, i.e., from cores of order $1$, and goes down vertically with the classic core decomposition. Based on Proposition~\ref{prp:conteinment}, it is possible to design a more efficient algorithm that
exploits also the ``horizontal containment'' relationships.

\begin{myexample}
Consider core  $C_{1,[0,2]}$  in Figure \ref{fig:searchspace}: by Proposition~\ref{prp:conteinment} it holds that it is a subset of both $C_{1,[0,1]}$ and $C_{1,[1,2]}$. Therefore, to compute $C_{1,[0,2]}$, instead of starting from the whole $V$, one can start from $C_{1,[0,1]} \cap C_{1,[1,2]}$. Starting from a much smaller set of vertices can provide a substantial speed-up to the whole computation.
\end{myexample}

This observation, although simple,  produces a speed-up of orders of magnitude as we will empirically show in Section~\ref{sec:experiments}.
The next straightforward corollary of Proposition~\ref{prp:conteinment} states that, not only $C_{1,[0,2]} \subseteq C_{1,[0,1]} \cap C_{1,[1,2]}$, but this is the best one can get, meaning that
intersecting these two span-cores is equivalent to intersecting all span-cores structurally containing $C_{1,[0,2]}$.

\begin{mycorollary}\label{cor:cor1ofProp1}
Given a temporal graph $G = (V, T, \tau)$, and a temporal interval $\Delta = [t_s,t_e] \sqsubseteq T$, let $\Delta_+ = [\min\{t_s+1,t_e\}, t_e]$ and $\Delta_- = [t_s, \max\{t_e-1,t_s\}]$.
It holds that $$C_{1,\Delta} \ \subseteq \ (C_{1,\Delta_+} \cap C_{1,\Delta_-}) \ = \ \bigcap_{\Delta' \sqsubseteq \Delta} C_{1, \Delta'} .$$
\end{mycorollary}


\begin{myexample}
Consider again $C_{1,[0,2]}$  in Figure \ref{fig:searchspace}:  Proposition~\ref{prp:conteinment} states that it is a subset of $C_{1,[0,0]}, C_{1,[0,1]},C_{1,[1,1]},C_{1,[1,2]},C_{1,[2,2]}$.
Corollary~\ref{cor:cor1ofProp1} suggests that there is no need to intersect them all, but only $C_{1,[0,1]}$ and $C_{1,[1,2]}$: in fact, $C_{1,[0,1]} \subseteq C_{1,[0,0]} \cap C_{1,[1,1]}$ and
$C_{1,[1,2]} \subseteq C_{1,[1,1]} \cap C_{1,[2,2]}$.
\end{myexample}


\begin{algorithm}[t]
\DontPrintSemicolon
\small
\KwIn{A temporal graph $G=(V,T,\tau)$.}
\KwOut{The set $\coresset$ of all \spancores of $G$.}
$\coresset \leftarrow \emptyset$; \ \ $Q \leftarrow \emptyset$; \ \ $\mathcal{A} \leftarrow \emptyset$\;
\ForAll{$t \in T$}
{\label{line:decomposition:init:start}
	enqueue $[t,t]$ to $Q$; \ \ $\mathcal{A}[t,t] \leftarrow V$\;\label{line:decomposition:init:end}
}
\While{$Q \neq \emptyset$}
{\label{line:decomposition:whilestart}
	dequeue $\Delta = [t_s, t_e]$ from $Q$\;
	$E_{\Delta}[\mathcal{A}[\Delta]] \leftarrow  \{(u,v) \in E_{\Delta} \mid u \in \mathcal{A}[\Delta], v \in \mathcal{A}[\Delta]\}$\;\label{line:decomposition:subgraph}

	\If{$|E_{\Delta}[\mathcal{A}[\Delta]]| > 0$}
	{
		$\coressetdelta \leftarrow $ \textsf{core-decomposition}$(\mathcal{A}[\Delta],E_{\Delta}[\mathcal{A}[\Delta]])$\;\label{line:decomposition:deltacores}				
		$\coresset \leftarrow \coresset \cup \coressetdelta$\; \label{line:decomposition:solution}
		$\Delta_1 = [\max\{t_s-1,0\}, t_e]$; \ \  $\Delta_2 = [t_s, \min\{t_e+1,t_{max}\}]$\; \label{line:decomposition:fathers}	 
		\ForAll{$\Delta' \in \{\Delta_1, \Delta_2\} \mid \Delta' \neq \Delta$}
		{
			\uIf{$\mathcal{A}[\Delta'] \neq \textsc{null}$} 
			{\label{line:decomposition:alreadyinit}	
				$\mathcal{A}[\Delta'] \leftarrow \mathcal{A}[\Delta'] \cap C_{1,\Delta}$\; \label{line:decomposition:second}	
				enqueue $\Delta'$ to $Q$\; \label{line:decomposition:enqueue}
			}	
			\Else
			{
				$\mathcal{A}[\Delta'] \leftarrow C_{1,\Delta}$\;\label{line:decomposition:first}
			}
		}
	}
}
\caption{\cores}\label{alg:decomposition}
\end{algorithm}

The main idea behind our efficient \cores\ algorithm (whose pseudocode is given as Algorithm~\ref{alg:decomposition}) is to generate temporal intervals of increasing size (starting from size one) and, for each $\Delta$ of width larger than one, to start the core decomposition from $(C_{1,\Delta_+} \cap C_{1,\Delta_-})$, i.e., the smallest intersection of cores containing $C_{1,\Delta}$ (Corollary~\ref{cor:cor1ofProp1}).
The intervals to be processed are added to queue $Q$, which is initialized with the intervals of size one (Lines~\ref{line:decomposition:init:start}--\ref{line:decomposition:init:end}): these are the only intervals for which no other interval can be used to reduce the set of vertices from which start the core decomposition, thus it has to be initialized with the whole vertex set $V$.
The algorithm utilizes a map $\mathcal{A}$ that, given an interval $\Delta$, returns the set of vertices to be used as a starting set of the core decomposition on $\Delta$.
The algorithm processes all intervals stored in $Q$, until $Q$ has become empty (Lines~\ref{line:decomposition:whilestart}--\ref{line:decomposition:first}).
For every temporal interval $\Delta$ extracted from $Q$, the starting set of vertices is retrieved from $\mathcal{A}[\Delta]$  and the corresponding set of edges is identified (Line~\ref{line:decomposition:subgraph}). Unless this is empty, the  classic core-decomposition algorithm~\cite{batagelj2011fast} is invoked over $(\mathcal{A}[\Delta],E_{\Delta}[\mathcal{A}[\Delta]])$ (Line~\ref{line:decomposition:deltacores}) and its output (a set of span-cores of span $\Delta$) is added to the ultimate output set $\coresset$ (Line~\ref{line:decomposition:solution}).

Afterwards, the two intervals, denoted $\Delta_1$ and $\Delta_2$, for which $C_{1,\Delta}$ can be used to obtain the smallest intersections of cores containing them (Corollary~\ref{cor:cor1ofProp1}) are computed at Line~\ref{line:decomposition:fathers}.
For $\Delta_1$ (and analogously $\Delta_2$), we check whether $\mathcal{A}[\Delta_1]$ has already been initialized (Line~\ref{line:decomposition:alreadyinit}): this would mean that previously the other ``father'' (i.e., smallest containing core) of $C_{1,\Delta_1}$ has been computed, thus we can intersect $C_{1,\Delta}$ with $\mathcal{A}[\Delta_1]$ and enqueue $\Delta_1$ to be processed (Lines~\ref{line:decomposition:second}--\ref{line:decomposition:enqueue}). Instead, if $\mathcal{A}[\Delta_1]$ was not yet initialized, we initialize it with $C_{1,\Delta}$ (Line~\ref{line:decomposition:first}): in this case $\Delta_1$ is not enqueued because it still misses one father to be intersected before being ready for core decomposition.
This procedural update of $Q$ ensures that both fathers of every interval in $Q$ exist and have been previously computed, thus no a-posteriori verification is needed.

\begin{myexample}
Consider again the search space in Figure~\ref{fig:searchspace}.
Algorithm~\ref{alg:decomposition} first processes the intervals $[0,0],[1,1],[2,2],$ and $[3,3]$.
Then, it intersects  $C_{1,[0,0]}$ and $C_{1,[1,1]}$ to initialize $C_{1,[0,1]}$, intersects $C_{1,[1,1]}$ and $C_{1,[2,2]}$ to initialize $C_{1,[1,2]}$, and  intersects $C_{1,[2,2]}$ and $C_{1,[3,3]}$ to initialize $C_{1,[2,3]}$. Then, it continues with the intervals of size 3: it intersects $C_{1,[0,1]}$ and $C_{1,[1,2]}$ to initialize $C_{1,[0,2]}$ and so on.
\end{myexample}
The next theorem formally shows soundness and completeness of our \cores\ algorithm.
\begin{mytheorem}\label{th:correctnessAlg2}
Algorithm~\ref{alg:decomposition} is sound and complete for Problem~\ref{pbl:dececomposition}.
\end{mytheorem}
\begin{proof}
The algorithm generates and processes a subset of temporal intervals $\mathcal{X} \subseteq \{\Delta \mid \Delta \sqsubseteq T\}$.
For every interval $\Delta \subseteq \mathcal{X}$, it computes \emph{all} \spancores $\mathbf{C}_{\Delta} = \{C_{1,\Delta}, C_{2,\Delta}, \ldots, C_{k_{\Delta},\Delta}\}$ defined on $\Delta$ by means of the \textsf{core-decomposition} subroutine on the graph $(\mathcal{A}[\Delta],E_{\Delta}[\mathcal{A}[\Delta]])$.
The set of vertices $\mathcal{A}[\Delta]$ is equivalent to $(C_{1,\Delta_+} \cap C_{1,\Delta_-})$ because of Line~\ref{line:decomposition:second} (Corollary~\ref{cor:cor1ofProp1}) and the fact that $\Delta$ is enqueued (Line~\ref{line:decomposition:enqueue}) only when both fathers have been processed and the intersection done. The correctness of doing the classic core decomposition is guaranteed by Observation ~\ref{observation1}.

As for completeness, it suffices to show that the intervals $\Delta \notin \mathcal{X}$ that have not been processed by the algorithm do not yield any \spancore.
The algorithm generates all temporal intervals size by size, starting from those of size one and then going to larger sizes. This is done by maintaining the queue $Q$. As said above, an interval $\Delta$ is enqueued as soon as both   $C_{1,\Delta_+}$ and $C_{1,\Delta_-}$ have been processed. Thus, an interval $\Delta$ is not in $\mathcal{X}$ only if either $C_{1,\Delta_+}$ or $C_{1,\Delta_-}$ does not exist. In this case $C_{1,\Delta}$ and all other $C_{k,\Delta}$ do not exist as well.
\end{proof}

\noindent \textbf{Discussion.}
Algorithm~\ref{alg:decomposition} exploits the ``horizontal containment'' relationships only at the first level of the search space.
For a given $\Delta$, once the restricted starting set of vertices has been defined for $k = 1$, the traditional core decomposition is started to produce all the span-cores of span $\Delta$.
In other words, for $k > 1$ only the ``vertical containment'' is exploited.
Consider the span-core $C_{3,[1,2]}$ in Figure \ref{fig:searchspace}: we know that it is a subset of  $C_{2,[1,2]}$ (``vertical'' ) and of $C_{3,[1,1]}$ and $C_{3,[2,2]}$  (``horizontal'' ).
One could consider intersecting all these three span-cores before computing $C_{3,[1,2]}$.
We tested this alternative approach, but concluded that the overhead of computing intersections and data-structure maintenance was outweighing the benefit of starting from a smaller vertex set.


The worst-case time complexity of Algorithm~\ref{alg:decomposition} is equal to the na\"{\i}ve approach, however in practice it is orders of magnitude faster, as shown in Section~\ref{sec:experiments}.

\section{Computing Maximal Span-cores}
\label{sec:algorithms:maximal}

In this section we focus on Problem~\ref{pbl:maximal}: computing the \emph{maximal} \spancores of a temporal graph.

\spara{A filtering approach.}
As anticipated above, a straightforward way of solving this problem consists in filtering the \spancores computed during the execution of Algorithm~\ref{alg:decomposition}, so as to ultimately output only the maximal ones.
This can easily be accomplished by equipping Algorithm~\ref{alg:decomposition} with a data structure $\mathcal{M}$ that stores the \spancore of the highest order for every temporal interval $\Delta \sqsubseteq T$ that has been processed by the algorithm.
Moreover, at the storage of a \spancore $C_{k,\Delta}$ in $\mathcal{M}$, the \spancores previously stored in $\mathcal{M}$ for subintervals of the temporal interval $\Delta$ and with the same order $k$ are removed from $\mathcal{M}$.
This removal operation, together with the order in which \spancores are processed, ensures that $\mathcal{M}$ eventually contains only the maximal \spancores.

%

\spara{Efficient maximal-span-core finding.}
Our next goal is to design a more efficient algorithm that extracts maximal \spancores directly, without computing complete core decompositions,
passing over more peripheral ones, and without generating all temporal cores.
This is a quite challenging design principle, as it contrasts the intrinsic structural properties of core decomposition, based on which a core of order $k$ is usually computed from the core of order $k\!-\!1$, thus making the computation of the core of the highest order as hard as computing the overall decomposition.
Nevertheless, thanks to theoretical properties that relate the maximal \spancores to each other, in the temporal context such a challenge can be achieved.
In the following we discuss such properties in detail, by starting from a result that has already been discussed above, but only informally.

Consider the classic core decomposition in a standard (\mbox{non-temporal}) graph $G$ (Definition~\ref{def:kcores}) and
let $C_{k^*}[G]$ denote the \emph{innermost} core of $G$, i.e., the non-empty $k$-core of $G$ with the largest~$k$.

\begin{mylemma}\label{lemma1}
Given a temporal graph $G = (V,T,\tau)$, let $\imcores$ be the set of all maximal \spancores of $G$, and $\mathbf{C_{inner}} = \{ C_{k^*}[G_\Delta]  \mid \Delta \sqsubseteq T\}$ be the set of innermost cores of all graphs $G_\Delta$. It holds that $\imcores \subseteq \mathbf{C_{inner}}$.
\end{mylemma}
\begin{proof}
Every $C_{k,\Delta} \in \imcores$ is the innermost core of the non-temporal graph $G_\Delta$: else,
there would exist another core $C_{k',\Delta} \neq \emptyset$ with $k' > k$, implying that $C_{k,\Delta} \notin \imcores$.
\end{proof}

Lemma~\ref{lemma1} states that each maximal \spancore is an innermost core of a $G_\Delta$, for some temporal interval $\Delta \sqsubseteq T$.
Hence, there can exist at most one maximal \spancore for every $\Delta \sqsubseteq T$ (while an interval $\Delta$ may not yield any maximal \spancore).
The key question to design an efficient maximal-\spancore-mining algorithm thus becomes how to extract innermost cores of
the graphs $G_\Delta$ more efficiently than by computing the full core decompositions of all $G_\Delta$.
The answer to this question comes from the result stated in the next two lemmas (with Lemma~\ref{lemma2} being auxiliary to Lemma~\ref{lemma3}).

\begin{mylemma}\label{lemma2}
Given a temporal graph $G = (V,T,\tau)$, and three temporal intervals $\Delta = [t_s,t_e] \sqsubseteq T$, $\Delta' = [t_s\!-\!1,t_e] \sqsubseteq T$, and $\Delta'' = [t_s,t_e\!+\!1] \sqsubseteq T$.
The innermost core $C_{k^*}[G_\Delta]$ is a maximal span-core of $G$ if and only if $k^* > \max\{k',k''\}$ where $k'$ and $k''$ are the orders of the innermost cores of $G_{\Delta'}$ and $G_{\Delta''}$, respectively.
\end{mylemma}
\begin{proof}
The ``$\Rightarrow$'' part comes directly from the definition of maximal \spancore (Definition~\ref{def:maximal}): if $k^*$ were
not  larger than $\max\{k',k''\}$, then $C_{k^*}[G_\Delta]$ would be dominated by another \spancore both on the order and on the span (as both $\Delta'$ and $\Delta''$ are superintervals of $\Delta$).
For the ``$\Leftarrow$'' part, from Lemma~\ref{lemma1} and Proposition~\ref{prp:conteinment} it follows that $\max\{k',k''\}$ is an upper bound on the maximum order of a \spancore of a superinterval of $\Delta$.
Therefore, $k^* > \max\{k',k''\}$ implies that there cannot exist any other \spancore that dominates $C_{k^*}[G_\Delta]$ both on the order and on the span.
\end{proof}

\begin{mylemma}\label{lemma3}
Given $G$, $\Delta$,  $\Delta'$,  $\Delta''$,  $k'$, and $k''$ defined as in Lemma~\ref{lemma2}, let $\widetilde{V} = \{u \in V \mid \tdeg_{\Delta}(V, u) > \max\{k', k''\}\}$,
and let $C_{k^*}[G_\Delta[\widetilde{V}]]$  be the innermost core of $G_\Delta[\widetilde{V}]$.
If $k^* >  \max\{k', k''\}$, then $C_{k^*}[G_\Delta[\widetilde{V}]]$ is a maximal \spancore; otherwise, no maximal \spancore exists for $\Delta$.
\end{mylemma}
\begin{proof}
Lemma~\ref{lemma2} states that, to be recognized as a maximal \spancore, the innermost core of $G_{\Delta}$ should have order larger than $\max\{k', k''\}$.
This means that, if the innermost core of $G_{\Delta}$ is a maximal \spancore, all vertices $u \notin \widetilde{V}$ cannot be part of it.
Therefore, $G_{\Delta}$ yields a maximal \spancore only if the innermost core of subgraph  $G_{\Delta}[\widetilde{V}]$ has order $k^* >  \max\{k', k''\}$.
\end{proof}
Lemma~\ref{lemma3} provides the basis of our efficient method for extracting maximal \spancores.
Basically, it states that, to verify whether a certain temporal interval $\Delta = [t_s,t_e]$ yields a maximal \spancore (and, if so, compute it), there is no need to consider the whole graph $G_\Delta$, rather it suffices to start from a smaller subgraph, which is given by all vertices whose temporal degree is larger than the maximum between the orders of the innermost cores of intervals $\Delta' = [t_s\!-\!1,t_e]$ and $\Delta'' = [t_s,t_e\!+\!1]$.
This finding suggests a strategy that is opposite to the one used for computing the overall \spancore decomposition:
a \emph{top-down} strategy that processes temporal intervals starting from the larger ones.
Indeed, in addition to exploiting the result in Lemma~\ref{lemma3}, this way of exploring the temporal-interval space allows us to skip the computation of complete core decompositions of the whole ``singleton-interval'' graphs $\{G_{_{[t,t]}}\}_{t \in T}$, which may easily become a critical bottleneck, as they are the largest ones among the graphs induced by temporal intervals.

\begin{algorithm}[t]
\DontPrintSemicolon
\small
\KwIn{A temporal graph $G=(V,T,\tau)$.}
\KwOut{The set $\imcores$ of all maximal \spancores of $G$.}

$\imcores \leftarrow \emptyset$\;
$\mathcal{K}'[t] \gets 0$, $\forall t \in T$\;

\ForAll{$t_s \in [0, 1, \ldots, t_{max}]$}
{\label{line:imcores:extfor}
	$t^* \leftarrow \max\{ t_e \in [t_s, t_{max}] \mid E_{_{[t_s,t_e]}} \neq \emptyset\}$\; \label{line:imcores:t}
	$k'' \leftarrow 0$\;
	\ForAll{$t_e \in [t^*, t^*\!-\!1, \ldots, t_s]$}
	{\label{line:imcores:intfor} 
		$\Delta \leftarrow [t_s,t_e]$\; \label{line:imcores:delta}
		$lb\gets \max\{\mathcal{K}'[t_e],k''\}$\; \label{line:imcores:lb}
		$V_{lb} \leftarrow \{u \in V \mid \tdeg_{\Delta}(V,u) > lb\}$\; \label{line:imcores:V}
		$E_\Delta[V_{lb}] \gets \{(u,v) \in E_{\Delta} \mid u \in V_{lb}, v \in V_{lb}\}$\; \label{line:imcores:E}
		$C \leftarrow $ \innermost $(V_{lb}, E_{\Delta}[V_{lb}])$\; \label{line:imcores:core}
		$k^* \leftarrow $ order of $C$\; \label{line:imcores:core2}
		\If{$k^* > lb$}
		{\label{line:imcores:updatesolution}
			$\imcores \leftarrow \imcores \cup \{C\}$\; \label{line:imcores:updatesolution2}
		}
		$k'' \gets \max\{k'', k^*\}$; \ $\mathcal{K}'[t_e] \gets \max\{\mathcal{K}'[t_e], k''\}$\; \label{line:imcores:updatek}
	}
}
\caption{\innermosts}\label{alg:imcores}
\end{algorithm}

\spara{The \innermosts\ algorithm.} 
Algorithm~\ref{alg:imcores} iterates over all timestamps $t_s \in T$ in \emph{increasing order} (Line~\ref{line:imcores:extfor}), and for each $t_s$ it first finds all the maximal span-cores that have span starting in $t_s$.
This way of proceeding \emph{ensures that a span-core that is recognized as maximal will not be later dominated by another span-core}.
Indeed, an interval $[t_s,t_e]$ can never be contained in another interval $[t_s',t_e']$ with $t_s < t_s'$.
For a given $t_s$, all maximal span-cores are computed as follows.
First, the maximum timestamp $\geq t_s$ such that the corresponding edge set $E_{_{[t_s,t_e]}}$ is not empty is identified as $t^*$ (Line~\ref{line:imcores:t}).
Then, all intervals $\Delta = [t_s, t_e]$ are considered one by one in \emph{decreasing order} of $t_e$ (Lines~\ref{line:imcores:intfor}--\ref{line:imcores:delta}): this again \emph{guarantees that a span-core that is recognized as maximal will not be later dominated by another span-core, as the intervals are processed from the largest to the smallest}.
At each iteration of the internal cycle, the algorithm resorts to Lemma~\ref{lemma3} and computes the lower bound $lb$ on the order of the innermost core of $G_{\Delta}$ to be recognized as maximal, by taking the maximum between $\mathcal{K}'[t_e]$ and $k''$ (Line~\ref{line:imcores:lb}).
$\mathcal{K}'$ is a map that maintains, for every timestamp $t \in [t_s, t^*]$, the order of the innermost core of graph $G_{\Delta'}$, where $\Delta' = [t_s\!-\!1, t]$ (i.e., $\mathcal{K}'[t]$ stores what in Lemmas~\ref{lemma2}--\ref{lemma3} is denoted as $k'$).
Whereas $k''$ stores the order of the innermost core of $G_{\Delta''}$, where $\Delta'' = [t_s, t_e+1]$.
Afterwards, the sets of vertices $V_{lb}$ and of edges $E_{\Delta}[V_{lb}]$ that comply with this lower-bound constraint are built (Lines~\ref{line:imcores:V}--\ref{line:imcores:E}), and the innermost core of the subgraph $(V_{lb}, E_{\Delta}[V_{lb}])$ is extracted (Lines~\ref{line:imcores:core}--\ref{line:imcores:core2}).
Ultimately, based again on Lemma~\ref{lemma3}, such a core is added to the output set of maximal \spancores only if its order is actually larger than $lb$ (Lines~\ref{line:imcores:updatesolution}--\ref{line:imcores:updatesolution2}), and the values of $k''$ and $\mathcal{K}'[t_e]$ are updated (Line~\ref{line:imcores:updatek}).
Specifically, note that the order $k^*$ of core $C$ may in principle be less than $k''$, as $C$ is extracted from a subgraph of $G_{\Delta}$.
If this happens, it means that the actual order of the innermost core of $G_{\Delta}$ is equal to $k''$.
This motivates the update rules (and their order) reported in Line~\ref{line:imcores:updatek}.


\begin{mytheorem}\label{th:correctnessAlg3}
Algorithm~\ref{alg:imcores} is sound and complete for Problem~\ref{pbl:maximal}.
\end{mytheorem}
\begin{proof}
The algorithm processes all temporal intervals $\Delta \sqsubseteq T$ yielding a non-empty edge set $E_\Delta$, in an order such that no interval is processed before one of its superintervals: this guarantees that a span-core recognized as maximal will not be dominated by another span-core found later on. For every $\Delta$ it extracts a core $C$ that is used as a proxy of the innermost core of graph $G_{\Delta}$.
$C$ is added to the output set $\imcores$
only if Lemma~\ref{lemma3} recognizes it as a maximal \spancore, otherwise it is discarded.
This proves the soundness of the algorithm.
Completeness follows from Lemma~\ref{lemma1}, which states that to extract all maximal \spancores it suffices to focus on the innermost cores of graphs $\{G_{\Delta} \mid \Delta \sqsubseteq T\}$, and Lemma~\ref{lemma3} again, which states the condition for a proxy core $C$ to be safely  discarded because it is a non-maximal \spancore.
\end{proof}

\spara{Discussion.}
The worst-case time complexity of Algorithm~\ref{alg:imcores} is the same as the algorithm for computing the overall \spancore decomposition, i.e., $\bigO(|T|^2 \times |E|)$.
It is worth mentioning that it is not possible to do better than this, as the output itself is potentially quadratic in $|T|$.
However, as we will show in Section~\ref{sec:experiments}, the proposed algorithm is in practice much more efficient than computing the overall \spancore decomposition and filtering out the non-maximal span-cores as, in this case, we avoid the visit of portions of the \spancore search space and the computations are run over subgraphs of reduced dimensions.

To conclude, we discuss how the crucial operation of building the subgraph $(V_{lb}, E_{\Delta}[V_{lb}])$ may be carried out efficiently in terms of both time and space.
Consider a fixed timestamp $t_s \in [0, \ldots, t_{max}]$.
The following reasoning holds for every $t_s$.
Let $E^-(t_e) = E_{_{[t_s,t_e]}} \setminus E_{_{[t_s,t_e\!+1]}}$ be the set of edges that are in $E_{_{[t_s,t_e]}}$ but not in $E_{_{[t_s,t_e\!+1]}}$, for
 $t_e \in [t_s, \ldots, t^*\!-1]$.
As a first general step, for each $t_s$, we compute and store \emph{all} edge sets $\{E^-(t_e)\}_{t_e \in [t_s, t^*\!-1]}$.
These operations can be accomplished in $\mathcal{O}(|T| \times |E|)$ overall time, because every $E^-(t_e\!)$ can be computed incrementally from $E_{_{[t_s,t_e]}}$ as $E^-(t_e) = \{(u,v) \in E_{_{[t_s,t_e]}} \mid \tau(u,v,t_e\!+\!1) = 0\}$.
Moreover, for any timestamp $t_e$, we keep a map $\mathcal{D}$ storing all vertices of $G_{_{[t_s,t_e]}}$ organized by degree.
Specifically, the set $\mathcal{D}[k]$ contains all vertices having degree $> k$ in  $G_{_{[t_s,t_e]}}$.
Every vertex in $\mathcal{D}$  is thus replicated a number of times equal to its degree.
This way, the overall space taken by $\mathcal{D}$ is $\bigO(|E|)$, i.e., as much space as $G$.
$\mathcal{D}$ is initialized as empty (when $t_e = t^*$) and repeatedly augmented as $t_e$ decreases, by a linear scan of the various $E^-(t_e)$.
The overall filling of $\mathcal{D}$ (for all $t_e$) therefore takes $\mathcal{O}(|T| \times |E|)$ time.
Then, the desired $V_{lb}$ can be computed in constant time simply as $V_{lb} = \mathcal{D}[lb]$.

As for $E_{\Delta}[V_{lb}]$, for any $t_e$, we first reconstruct $E_{_{[t_s,t_e]}}$ as $E_{_{[t_s,t_e+\!1]}} \cup E^-(t_e)$, having previously computed $E_{_{[t_s,t_e+\!1]}}$.
Note that storing all $E^-(t_e)$ takes $\bigO(|E|)$ space.
That is why we store all $E^-(t_e)$ and reconstruct $E_{_{[t_s,t_e]}}$ afterward (instead of storing the latter, which would take $\bigO(|T| \times |E|)$ space).
$E_{\Delta}[V_{lb}]$ is ultimately derived by a linear scan of $E_{_{[t_s,t_e]}}$, taking all edges in $E_{_{[t_s,t_e]}}$ having both endpoints in $V_{lb}$.
This way, the step of building $E_{\Delta}[V_{lb}]$  for all $t_e$ takes again $\mathcal{O}(|T| \times |E|)$ overall time.


\section{Experiments}
\label{sec:experiments}


In this section we present a performance comparison of our algorithms, as well as a characterization of span-cores extracted.

\noindent\textbf{Datasets.}
We use eight real-world datasets recording timestamped interactions between entities.\footnote{All datasets are made available by the KONECT Project (\url{http://konect.cc}), except for \textsf{StackOverflow} which is part of the SNAP Repository (\url{http://snap.stanford.edu}).}
For each dataset we select a window size to define a discrete time domain, composed of contiguous timestamps of the same duration, and build the corresponding temporal graph.
If multiple interactions occur between two entities during the same discrete timestamp, they are counted as one.
The characteristics of the resulting temporal graphs, along with the selected window sizes (in days), are reported in Table~\ref{tab:datasets}.

\textsf{ProsperLoans} represents the network of loans between the users of Prosper, a marketplace of loans between privates.
\textsf{Last.fm} records the co-listening activity of the Last.fm streaming platform: an edge exists between two users if they listened to songs of the same band within the same discrete timestamp.
\textsf{WikiTalk} is the communication network of the English Wikipedia.
\textsf{DBLP} is the co-authorship network of the authors of scientific papers from the DBLP computer science bibliography.
\textsf{StackOverflow} includes the answer-to-question interactions on the stack exchange of the stackoverflow.com website.
\textsf{Wikipedia} connects users of the Italian Wikipedia that co-edited a page during the same discrete timestamp.
Finally, for both \textsf{Amazon} and \textsf{Epinions}, vertices are users and edges represent the rating of at least one common item within the same discrete timestamp.

\spara{Implementation.}
All methods are implemented in Python (v. 2.7.12) and compiled by Cython.
The experiments run on a machine equipped with Intel Xeon CPU at 2.1GHz and 64GB RAM.

\noindent\textbf{Reproducibility.} Our code is available at \href{https://goo.gl/4WmrPc}{goo.gl/4WmrPc}.

\begin{table}[t!]
\centering
\small
\caption{\small \mbox{Temporal graphs used in the experiments.}}
\vspace{-3mm}
\label{tab:datasets}

\begin{tabular}{c|ccccccc}
\multicolumn{1}{c}{} & \multicolumn{1}{c}{} & \multicolumn{1}{c}{} & \multicolumn{1}{c}{} & \multicolumn{1}{c}{window} & \multicolumn{1}{c}{} \\
\multicolumn{1}{c}{dataset} & \multicolumn{1}{c}{$|V|$} & \multicolumn{1}{c}{$|E|$} & \multicolumn{1}{c}{$|T|$} & \multicolumn{1}{c}{size (days)} & \multicolumn{1}{c}{domain} \\
\hline
\textsf{ProsperLoans} & $89$k & $3$M & $307$ & $7$ & economic \\
\textsf{Last.fm} & $992$ & $4$M & $77$ & $21$ & co-listening \\
\textsf{WikiTalk} & $2$M & $10$M & $192$ & $28$ & communication \\
\textsf{DBLP} & $1$M & $11$M & $80$ & $366$ & co-authorship \\
\textsf{StackOverflow} & $2$M & $16$M & $51$ & $56$ & question answering \\
\textsf{Wikipedia} & $343$k & $18$M & $101$ & $56$ & co-editing \\
\textsf{Amazon} & $2$M & $22$M & $115$ & $28$ & co-rating \\
\textsf{Epinions} & $120$k & $33$M & $25$ & $21$ & co-rating \\
\hline
\end{tabular}
\vspace{1mm}
\end{table}

\subsection{\Spancore decomposition}
We compare the two methods to compute a complete decomposition
described in Section~\ref{sec:algorithms}, i.e.,
the baseline \baseline\ and the proposed \cores, in terms of execution time, memory, and total number of vertices input to the $\textsf{core-decomposition}$ subroutine. We report these measures, together with the numbers of \spancores and maximal \spancores of each dataset, in Table~\ref{tab:evaluation}.

In terms of execution time, \cores\ considerably outperforms \baseline\ in all datasets, achieving a speed-up from $2.1$ up to two orders of magnitude.
The speed-up is explained by the number of vertices processed by the $\textsf{core-decomposition}$ subroutine, which is the most time-consuming step of the algorithms albeit linear in the size of the input subgraph.
The difference of this quantity between \cores\ and \baseline\ reaches an order of magnitude in the \textsf{WikiTalk}, \textsf{Wikipedia}, and \textsf{Epinions} dataset, confirming the effectiveness of the ``horizontal containment'' relationships.
The memory required by the two procedures is comparable in all cases since the largest structures needed in memory are the temporal graph itself and the set $\coresset$ of all \spancores.

\subsection{Maximal \spancores}
We compare our \innermosts\ algorithm to the na\"ive approach, described ad the beginning of Section~\ref{sec:algorithms:maximal}, based on running the \cores\ algorithm and filtering out the non-maximal \spancores, which we refer to as \baselineinnermosts.
The results are again reported in Table~\ref{tab:evaluation}.

\baselineinnermosts\ behaves very similarly to \cores: they only differ for the filtering mechanism which requires a few additional seconds in most cases.
\innermosts\ is much faster than \baselineinnermosts\ for all datasets, with a speed-up from $1.3$ for the \textsf{Epinions} dataset to $9.4$ for the \textsf{WikiTalk} dataset.
Except for the datasets \textsf{Last.fm} and \textsf{Epinions}, the difference in terms of number of processed vertices is between two and three orders of magnitude, proving the advantages of the top-down strategy of \innermosts, which avoids the visit of portions of the \spancore search space and handles the overhead of reconstructing graphs, i.e., $(V_{lb}, E_{\Delta}[V_{lb}])$, efficiently.
Finally, the memory requirements of the two methods are comparable for all datasets.

\begin{table}
\setlength{\tabcolsep}{3pt}
\centering
\scriptsize
\caption{\small Evaluation of the proposed algorithms: number of output \spancores, time, memory, and number of processed vertices.}
\vspace{-3mm}
\label{tab:evaluation}

\begin{tabular}{c|c|c|ccc}
\multicolumn{1}{c}{} & \multicolumn{1}{c}{} & \multicolumn{1}{c}{\# output} & \multicolumn{1}{c}{time} & \multicolumn{1}{c}{memory} & \multicolumn{1}{c}{\# processed} \\
\multicolumn{1}{c}{dataset} & \multicolumn{1}{c}{method} & \multicolumn{1}{c}{\spancores} & \multicolumn{1}{c}{(s)} & \multicolumn{1}{c}{(GB)} & \multicolumn{1}{c}{vertices} \\
\hline
 \multirow{4}{*}{\textsf{ProsperLoans}} & \baseline & \multirow{2}{*}{$4\,273$} & $101$ & $2$ & $55$M \\
 & \cores & & $46$ & $2$ & $27$M \\ \cline{2-6}
 & \baselineinnermosts &  \multirow{2}{*}{$293$} & $48$ & $2$ & $27$M \\
 & \innermosts & & $8$ & $2$ & $980$k \\
\hline
\multirow{4}{*}{\textsf{Last.fm}} & \baseline &  \multirow{2}{*}{$126\,819$} & $707$ & $0.5$ & $2$M \\
 & \cores & & $199$ & $0.5$ & $531$k \\ \cline{2-6}
 & \baselineinnermosts &  \multirow{2}{*}{$1\,670$} & $202$ & $0.5$ & $531$k \\
 & \innermosts & & $57$ & $0.5$ & $271$k \\
\hline
\multirow{4}{*}{\textsf{WikiTalk}} & \baseline &  \multirow{2}{*}{$19\,693$} & $322\,302$ & $36$ & $25$B \\
 & \cores & & $1\,084$ & $36$ & $555$M \\ \cline{2-6}
 & \baselineinnermosts &  \multirow{2}{*}{$632$} & $1\,194$ & $36$ & $555$M \\
 & \innermosts & & $126$ & $35$ & $2$M \\
\hline
\multirow{4}{*}{\textsf{DBLP}} & \baseline &  \multirow{2}{*}{$6\,135$} & $10\,506$ & $11$ & $1$B \\
 & \cores & & $278$ & $11$ & $150$M \\ \cline{2-6}
 & \baselineinnermosts &  \multirow{2}{*}{$268$} & $292$ & $11$ & $150$M  \\
 & \innermosts & & $116$ & $11$ & $620$k \\
\hline
\multirow{4}{*}{\textsf{StackOverflow}} & \baseline &  \multirow{2}{*}{$1\,238$} & $5\,360$ & $10$ & $1$B \\
 & \cores & & $245$ & $10$ & $127$M \\ \cline{2-6}
 & \baselineinnermosts &  \multirow{2}{*}{$129$} & $245$ & $10$ & $127$M \\
 & \innermosts & & $128$ & $10$ & $3$M \\
\hline
\multirow{4}{*}{\textsf{Wikipedia}} & \baseline &  \multirow{2}{*}{$125\,191$} & $17\,155$ & $4$ & $1$B \\
 & \cores & & $522$ & $4$ & $35$M \\ \cline{2-6}
 & \baselineinnermosts &  \multirow{2}{*}{$2\,147$} & $537$ & $4$ & $35$M \\
 & \innermosts & & $201$ & $4$ & $320$k \\
\hline
\multirow{4}{*}{\textsf{Amazon}} & \baseline &  \multirow{2}{*}{$29\,318$} & $10\,415$ & $18$ & $2$B \\
 & \cores & & $409$ & $18$ & $247$M \\ \cline{2-6}
 & \baselineinnermosts &  \multirow{2}{*}{$303$} & $580$ & $18$ & $247$M \\
 & \innermosts & & $123$ & $18$ & $688$k \\
\hline
\multirow{4}{*}{\textsf{Epinions}} & \baseline &  \multirow{2}{*}{$63\,111$} & $699$ & $4$ & $39$M \\
 & \cores & & $186$ & $4$ & $3$M \\ \cline{2-6}
 & \baselineinnermosts &  \multirow{2}{*}{$320$} & $201$ & $4$ & $3$M \\
 & \innermosts & & $154$ & $5$ & $129$k \\
\hline
\end{tabular}
\end{table}

\spara{Characterization.} 
%
We finally compare and characterize all \spancores against maximal \spancores.
At first, Table~\ref{tab:evaluation} shows that \spancores are at least one order of magnitude more numerous than maximal \spancores for all datasets, with the maximum difference of two orders of magnitude for the \textsf{Epinions} dataset.

In Figure~\ref{fig:coresvsmaximals_k} we show the number (top) and the average size (bottom) of \spancores and maximal \spancores as a function of the order $k$ for the \textsf{DBLP} and \textsf{Epinions} datasets.
For both datasets, the number of maximal \spancores is at least one order of magnitude lower than the total number of \spancores up to a quarter of the $k$ domain, where the \spancores are more numerous.
Instead, in the rest of the domain, they mostly coincide due to the maximality condition over $|\Delta|$.
The average size is also smaller for maximal \spancores, difference that wears thin when the gap between the numbers of \spancores and maximal \spancores starts decreasing since, for high values of $k$, most (or all) \spancores are maximal.

Figure~\ref{fig:coresvsmaximals_delta} shows a different picture when numbers and average sizes are shown as a function of the size of the span $|\Delta|$.
For both datasets, the number of \spancores and maximal \spancores decreases with, on average, a constant gap of one and two orders of magnitude, respectively, since the number of intervals decreases as $|\Delta|$ increases.
On the other hand, the behavior of the average size is quite different between the two datasets.
For the \textsf{DBLP} dataset, the average size of \spancores is much higher than the average size of maximal \spancores for low values of $|\Delta|$, then the difference decreases and
vanishes at the end of domain where a maximal \spancore of $|\Delta| = 37$ dominates all other \spancores of $|\Delta| \geq 20$.
Instead, for the \textsf{Epinions} dataset, the average size of all \spancores and 
maximal \spancores follows the same behavior, with a difference of less than
an order of magnitude, because the maximality condition over $k$ excludes the largest \spancores from the set of maximal \spancores.

\begin{figure}
\centerline{
\begin{tabular}{cc}
\includegraphics[width=0.5\columnwidth]{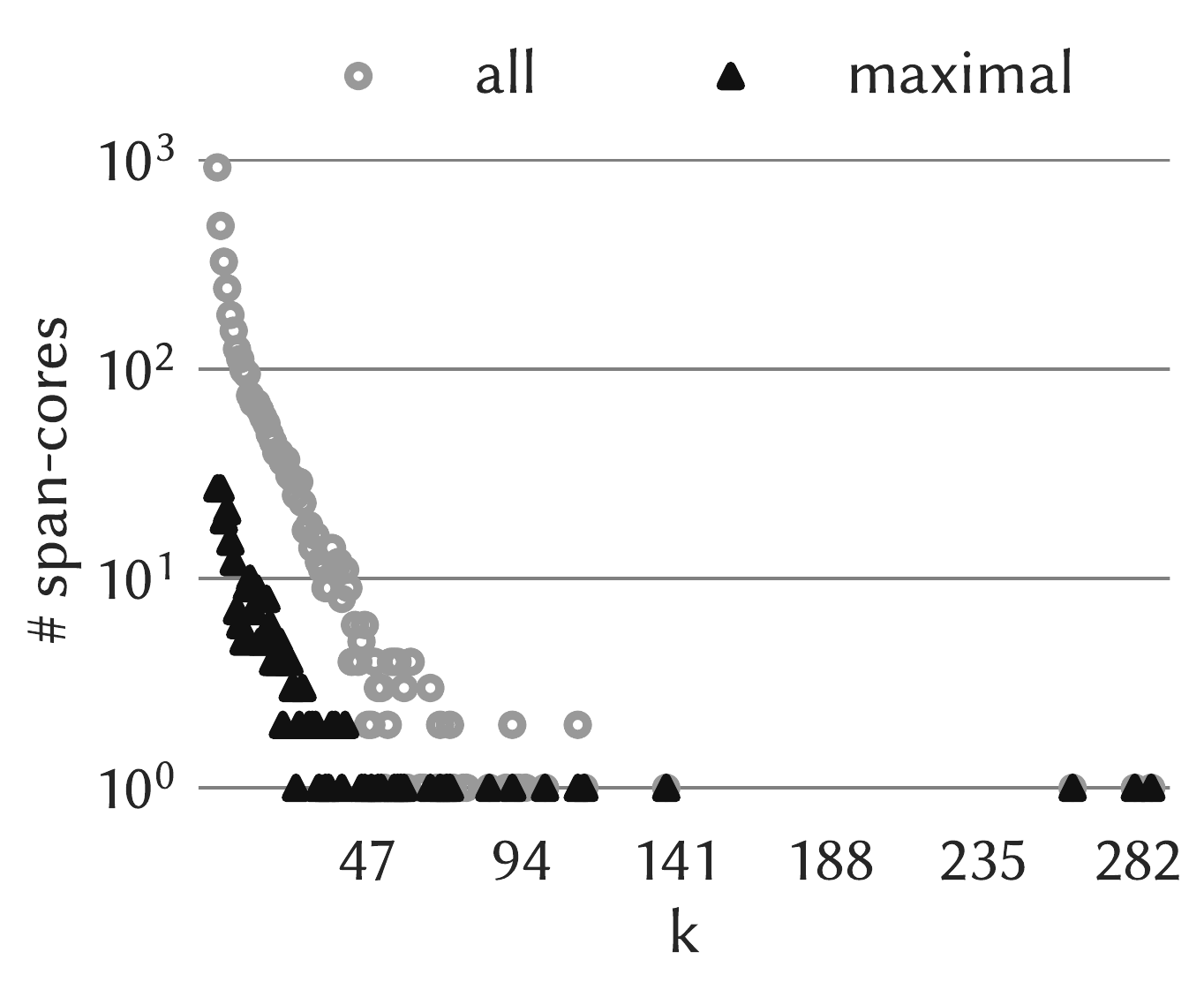} & \includegraphics[width=0.5\columnwidth]{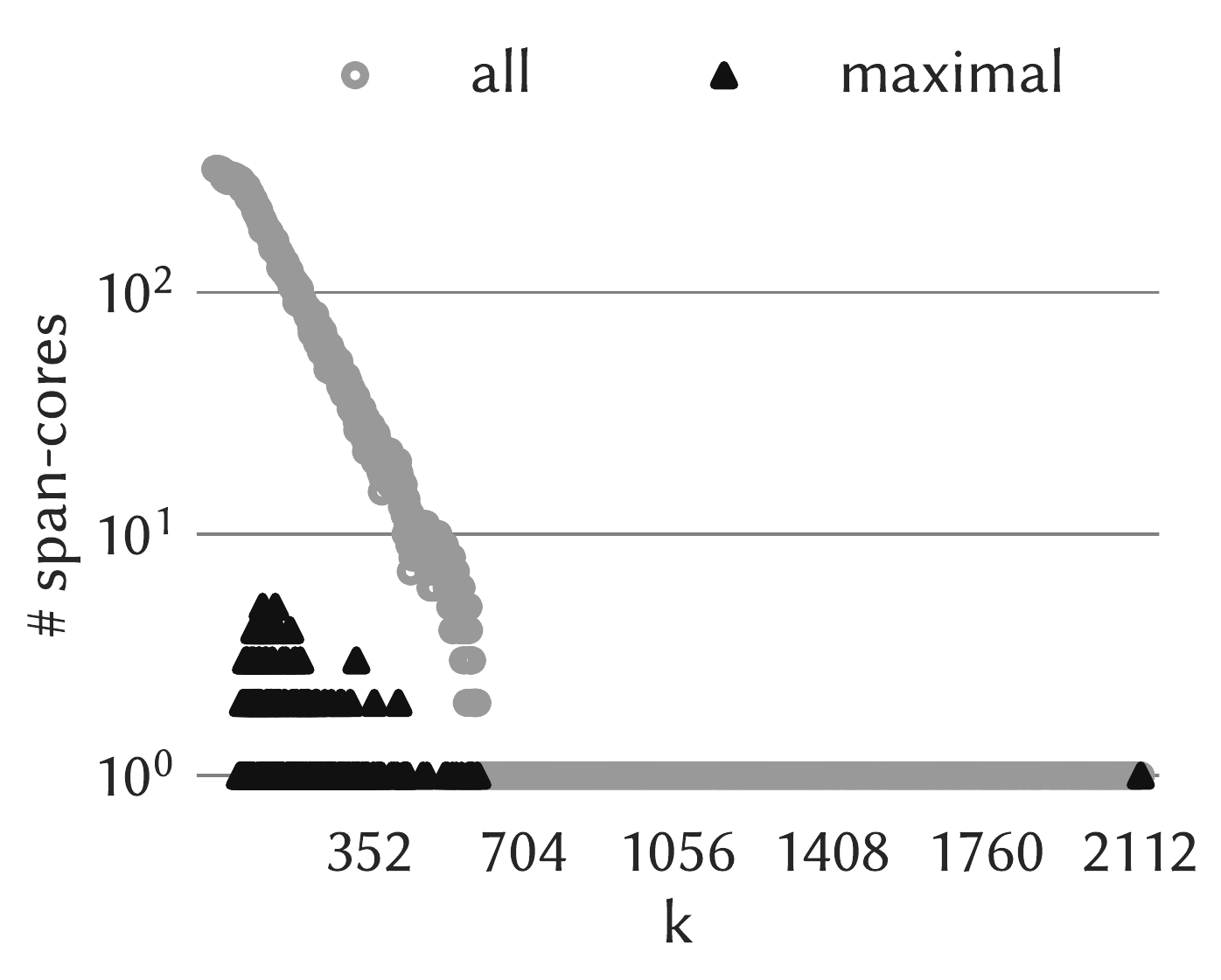} \vspace{-1mm}\\
\includegraphics[width=0.5\columnwidth]{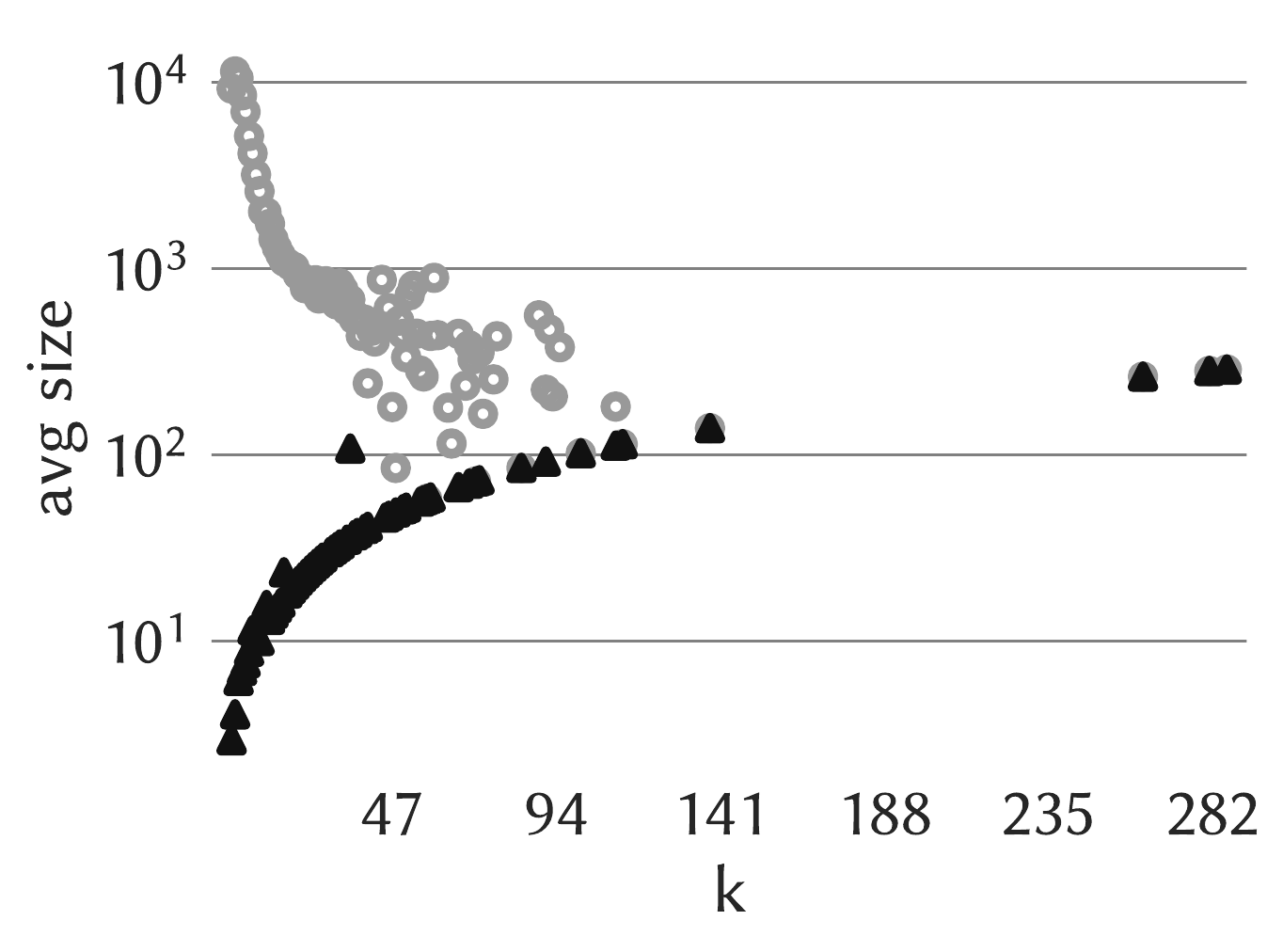} & \includegraphics[width=0.5\columnwidth]{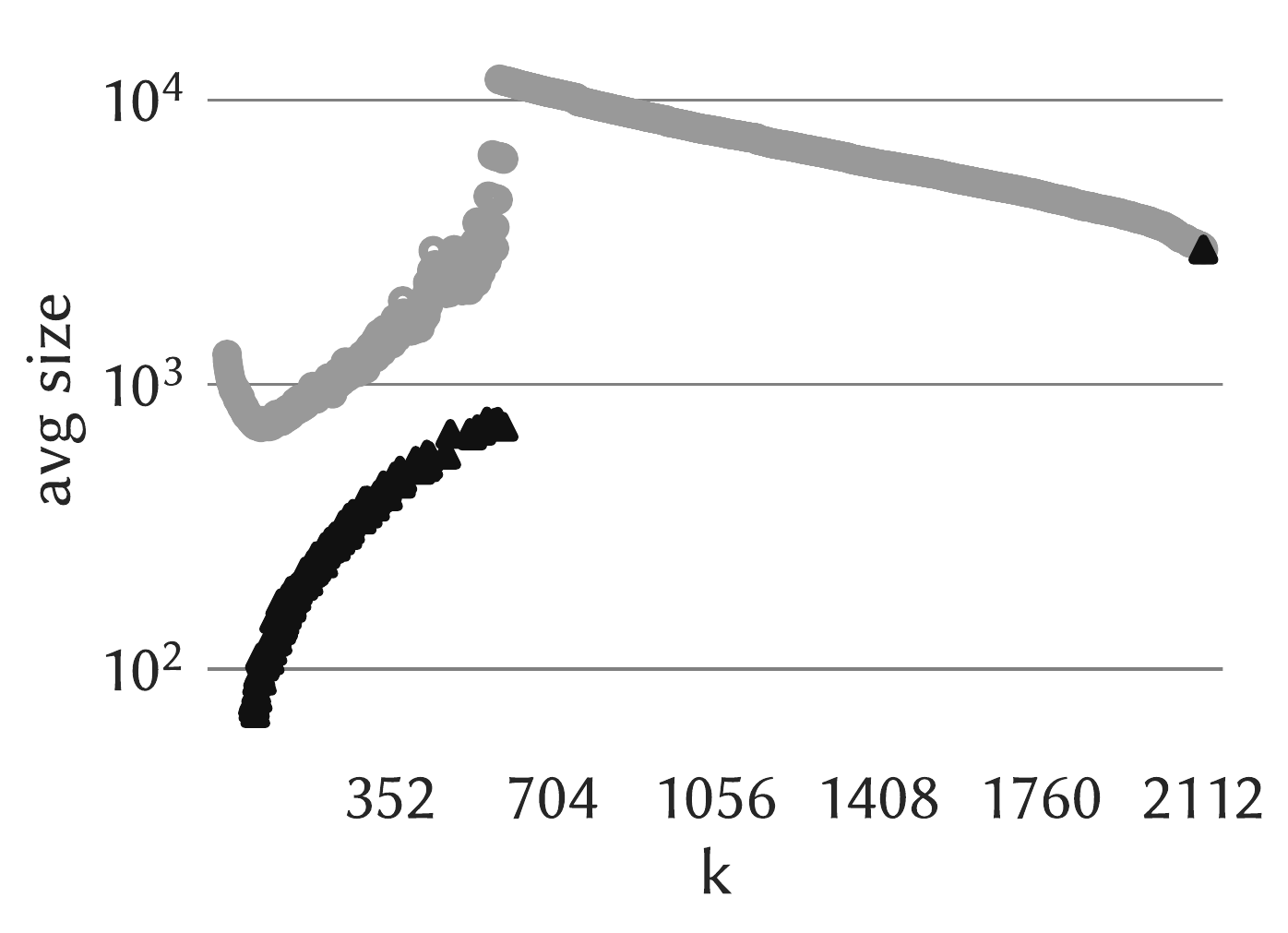} \vspace{-2mm}\\
\footnotesize{\textsf{DBLP}} & \footnotesize{\textsf{Epinions}}
\end{tabular}
}
\vspace{-3mm}

\caption{\small \label{fig:coresvsmaximals_k} Top plots: number of all \spancores and maximal \spancores ($y$ axis) as a function of the order $k$ ($x$ axis). Bottom plots: average size of all \spancores and maximal \spancores ($y$ axis) as a function of the order $k$ ($x$ axis).}
\end{figure}

\begin{figure}
\centerline{
\begin{tabular}{cc}
\includegraphics[width=0.5\columnwidth]{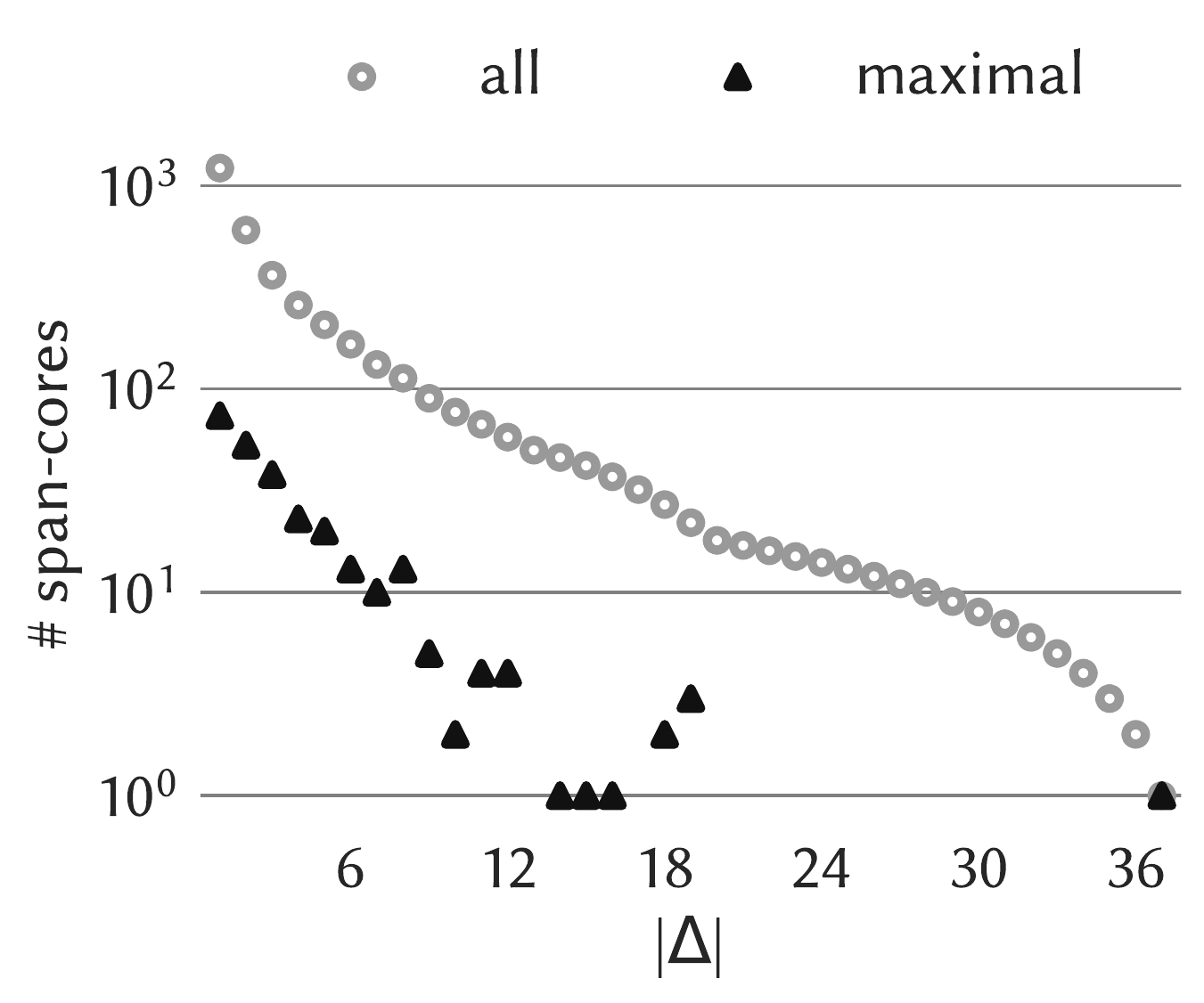} & \includegraphics[width=0.5\columnwidth]{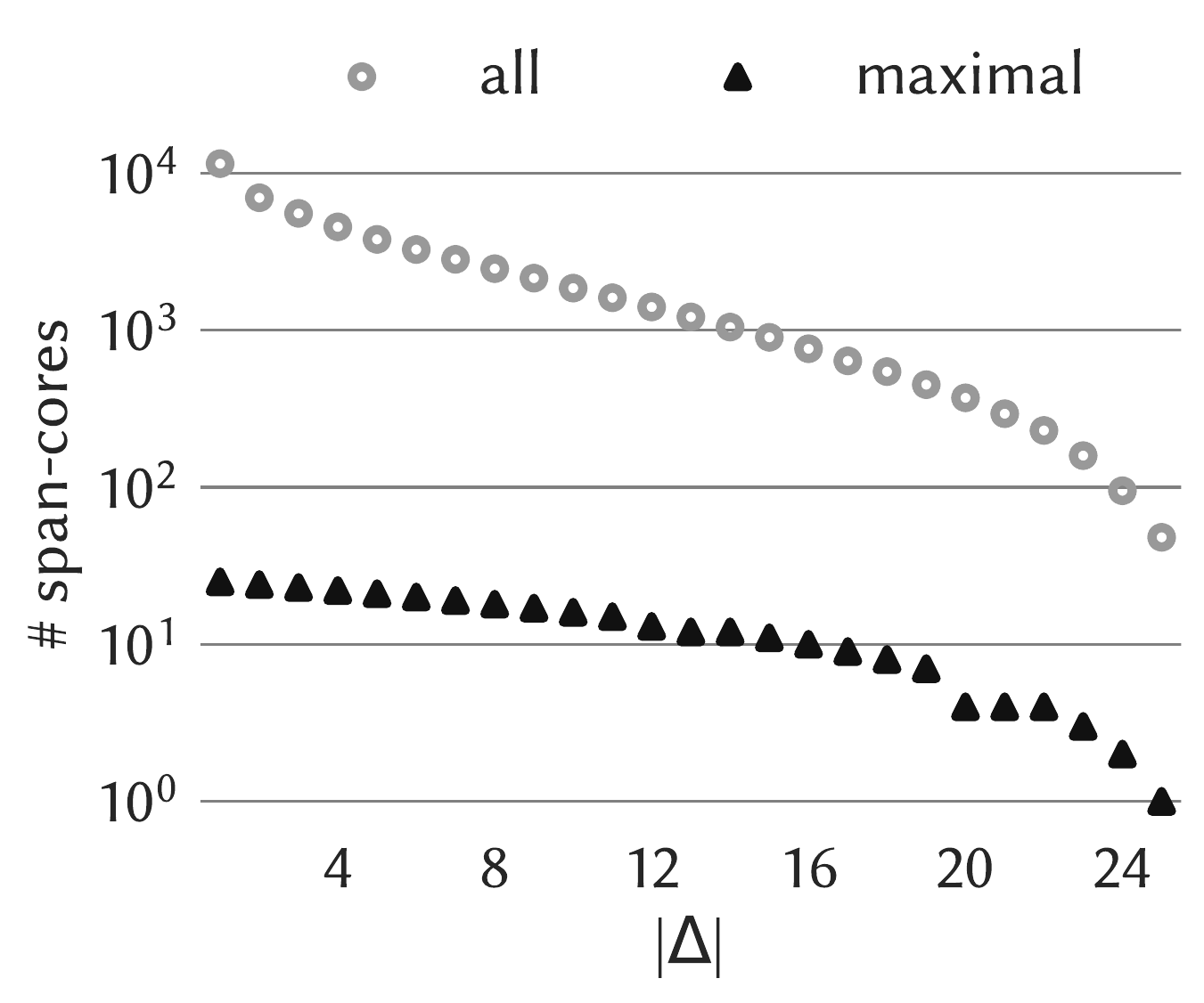} \vspace{-1mm}\\
\includegraphics[width=0.5\columnwidth]{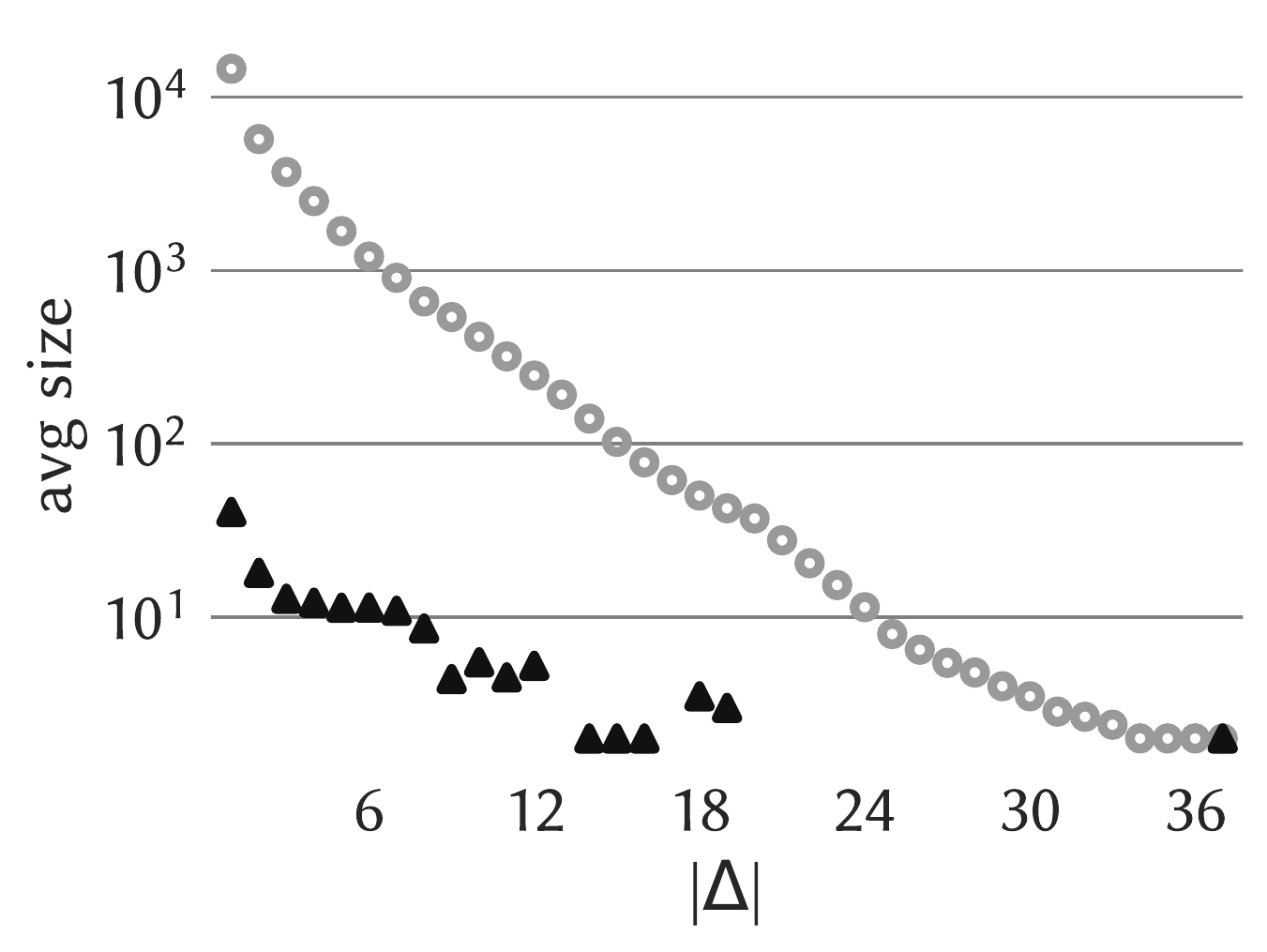} & \includegraphics[width=0.5\columnwidth]{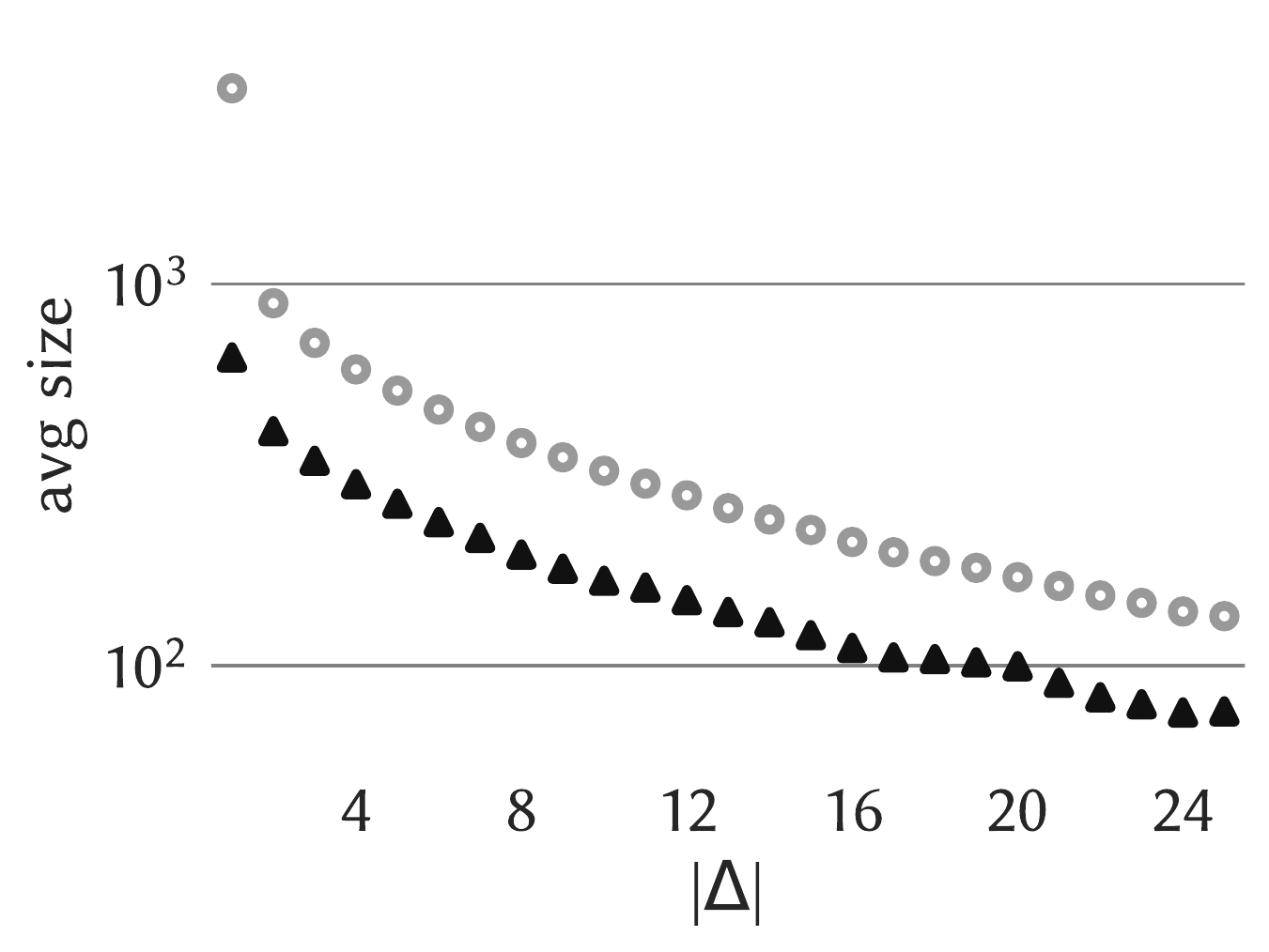} \vspace{-2mm}\\
\footnotesize{\textsf{DBLP}} & \footnotesize{\textsf{Epinions}}
\end{tabular}
}
\vspace{-3mm}

\caption{\small \label{fig:coresvsmaximals_delta} Top plots: number of all \spancores and maximal \spancores ($y$ axis) as a function of the size of the temporal span $|\Delta|$ ($x$ axis). Bottom plots: average size of all \spancores and maximal \spancores ($y$ axis) as a function of the size of the temporal span $|\Delta|$ ($x$ axis).}
\end{figure}

\section{Applications}
\label{sec:casestudies}

In this section we illustrate applications of (maximal) \spancores in the analysis of face-to-face interaction networks.
We use three datasets gathered by a proximity-sensing infrastructure with a resolution of $20$ seconds.
The first dataset, named \textsf{PrimarySchool}\footnote{Available at \href{http://www.sociopatterns.org}{sociopatterns.org}.\label{foot:school}}, contains the contact events between $242$
individuals ($232$ children and $10$ teachers) in a primary school in Lyon, during two days \cite{Stehle:2011}.
The \textsf{HighSchool}$^{\ref{foot:school}}$ dataset gives the interactions between students and teachers ($327$ individuals overall) of nine classes during five days in a high school in Marseilles \cite{Fournet:PLOS2015}.
Finally, the \textsf{HongKong} dataset describes the interactions of people in a primary school in Hong Kong for eleven consecutive days
\cite{sapienza2015detecting}.
The school population consists of $709$ children and $65$ teachers divided into thirty classes.
For all three datasets we use a window size of $5$ minutes and discard \spancores of $|\Delta| = 1$, i.e., having span of $5$ minutes, since they represent extremely short group interactions, not significant for our purposes.
On these datasets we show three types of interesting temporal patterns, i.e., social activities of groups of students within a school day, mixing of gender and class, and length of social interactions in groups.

\begin{figure}[h!t!]
\begin{tabular}{c}
\centerline{\includegraphics[width=1.05\columnwidth]{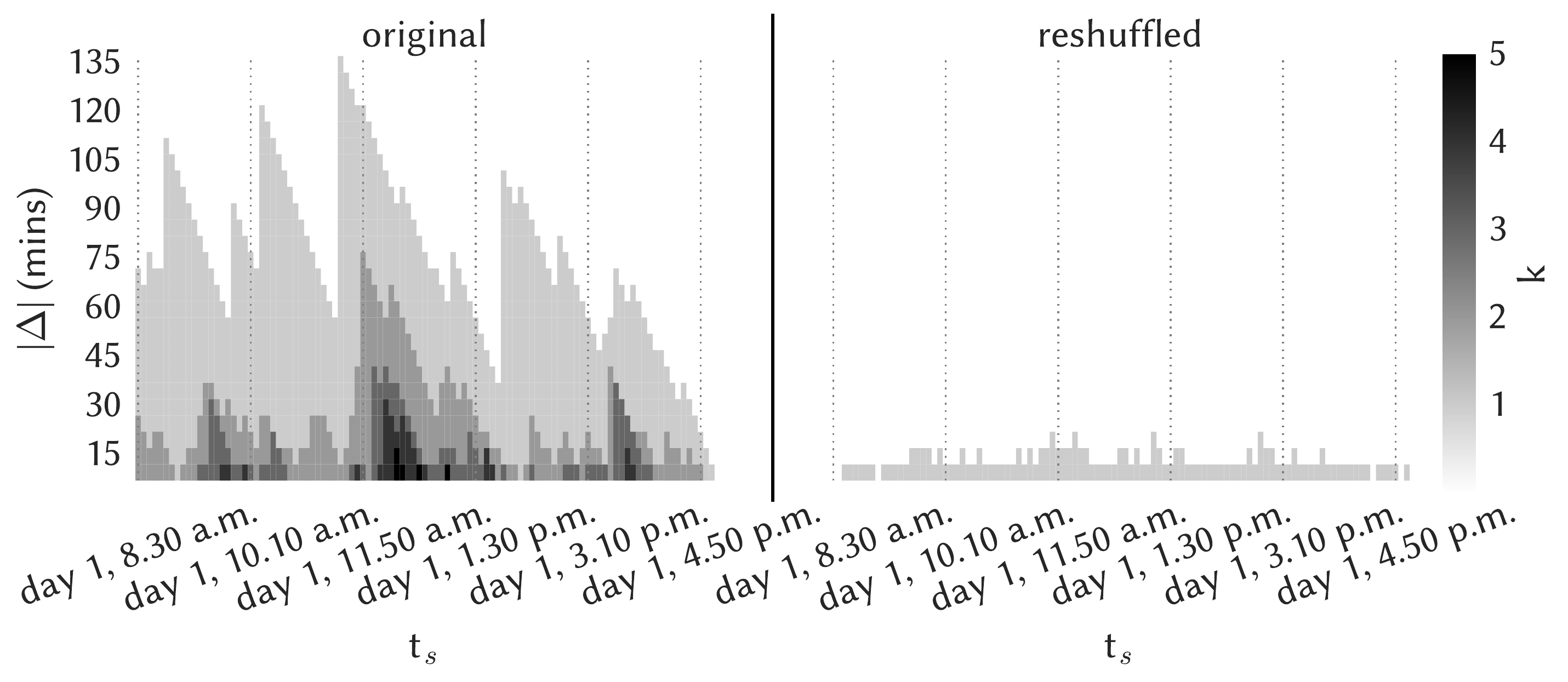}} \vspace{-3mm}\\
\footnotesize{\textsf{PrimarySchool}} \\
\centerline{\includegraphics[width=1.05\columnwidth]{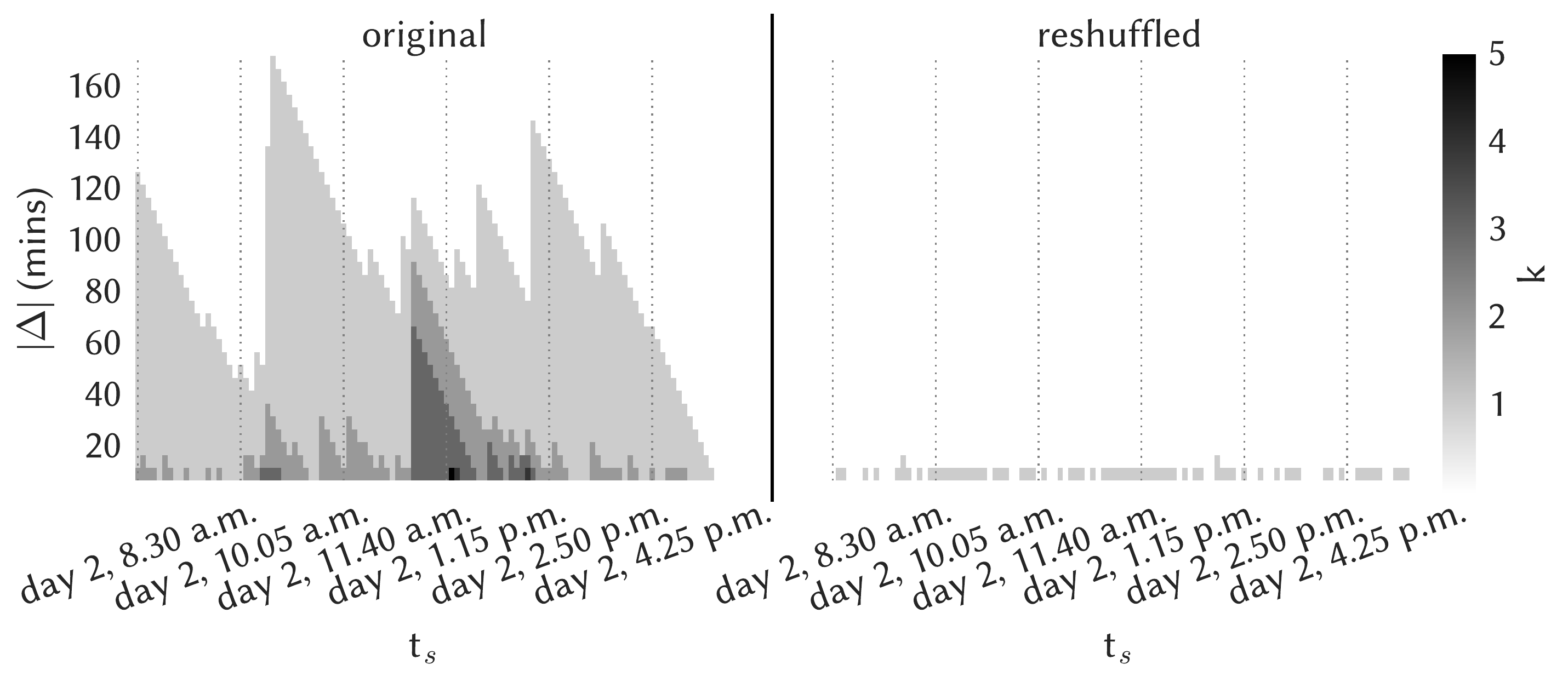}} \vspace{-3mm}\\
\footnotesize{\textsf{HighSchool}} \\
\end{tabular}
\vspace{-3mm}

\caption{\small \label{fig:temporalactivity} Temporal activity of a school day of the \textsf{PrimarySchool} and \textsf{HighSchool} datasets: the $x$ axis reports the hour of the day at which the span of a \spancore starts, the $y$ axis specifies the size of the span (in minutes), and the color scale shows the order $k$.
At a glance, it can be observed that the temporal structure of the \spancore decomposition detects time-evolving community structures in the original datasets (left plots) that completely disappears in the reshuffled datasets (right plots).}
\end{figure}

\subsection{Temporal patterns}
\spara{Temporal activity.}
We first show how \spancores\ yield a simple temporal analysis of social activities of groups of people within a school day.
The left side of Figure~\ref{fig:temporalactivity} reports colormaps of the order $k$ of the span-cores as a function of their starting time $t_s$ ($x$ axis) and of the size of their temporal span $|\Delta|$ ($y$ axis), for a school day of the \textsf{PrimarySchool} and \textsf{HighSchool} datasets.
Darker gray indicates \spancores of high order and slots located in the upper part of the plots refer to \spancores of long span.
In both datasets, fluctuations of $k$ and $|\Delta|$ are observed along the day, which can be related to school events.
Around $10$ a.m., the size of the span $|\Delta|$ reaches a local maximum in correspondence to the morning break, which means that students establish long-lasting interactions that hold beyond the break itself.
Moreover, when classes gather for the lunch break, the order $k$ reaches its maximum value since students tend to form larger and more cohesive groups.

In order to verify that these results are not trivially derived from the general temporal activity, as simply given by the number of interactions in each timestamp, we compare our findings to a null model.
At each timestamp of the temporal graphs, we reshuffle the edges by repeating the following operations, up to when all edges have been processed:
select at random two edges with no common vertices, e.g., $(u,v)$ and $(w,z)$, and transform them into $(u,z)$ and $(w,v)$.
This reshuffling preserves degree of each vertex in each timestamp and global activity (i.e., number of contacts per timestamp), but destroys correlations between edges of successive timestamps.
In the right side of Figure~\ref{fig:temporalactivity} we show the results of the temporal analysis described above for the reshuffled datasets.
In both, the values of $|\Delta|$ and $k$ reached are much smaller than in the original datasets.
The size of the span $|\Delta|$ is always shorter than $20$ minutes, while in the original datasets it is much longer, up to $170$ minutes, and the order $k$ is always equal to $1$, compared to the original maximum of $5$.
The time-evolving communities detected in the original datasets are completely lost after the reshuffling, where no temporal structure of the \spancores is observed.
This proves that the temporal schema of \spancore decomposition is not simply a consequence of the overall activity but that \spancores represent a concrete method to detect complex structures evolving in time.

\spara{Mixing patterns.}
We now show analysis of mixing patterns of students with respect to gender and class. Such metadata is indeed available for the individuals of the \textsf{PrimarySchool} dataset.
We define as \emph{gender purity} of a \spancore the fraction of individuals of the most represented gender within the \spancore.
\emph{Class purity} is analogously defined.
The left plot of Figure~\ref{fig:purity} reports the temporal evolution of gender and class purity during the first school day of the \textsf{PrimarySchool} dataset: at each timestamp $t$, the curves represent the average purities of the maximal \spancores spanning $t$.
During lessons, when students are in their own classes, class purity has naturally very high values, very close to $1$. Gender purity is instead rather low.
On the other hand, when students are gathered together, during the morning break at $10$ a.m. and the lunch break between $12$ a.m. and $2$ p.m., the situation is overturned: gender purity reaches large values while class purity drastically decreases.
This shows that primary school students group with individuals of the same class, disregarding the gender, only when they are forced by the schedule of the lessons, but prefer to interact with students of the same gender during breaks, in agreement with a previous study of the same dataset \cite{Stehle:2013}.

The right plot of Figure~\ref{fig:purity} shows the temporal evolution of gender and class purity with gender and class randomly reshuffled among individuals.
The two curves are more flat and the anti-correlation between them completely vanishes.
This testifies that the results on the original dataset are not simply due to the relative abundance of individuals of each type interacting at each time, but reflect genuine mixing patterns over time.


\begin{figure}
\centerline{\includegraphics[width=1.05\columnwidth]{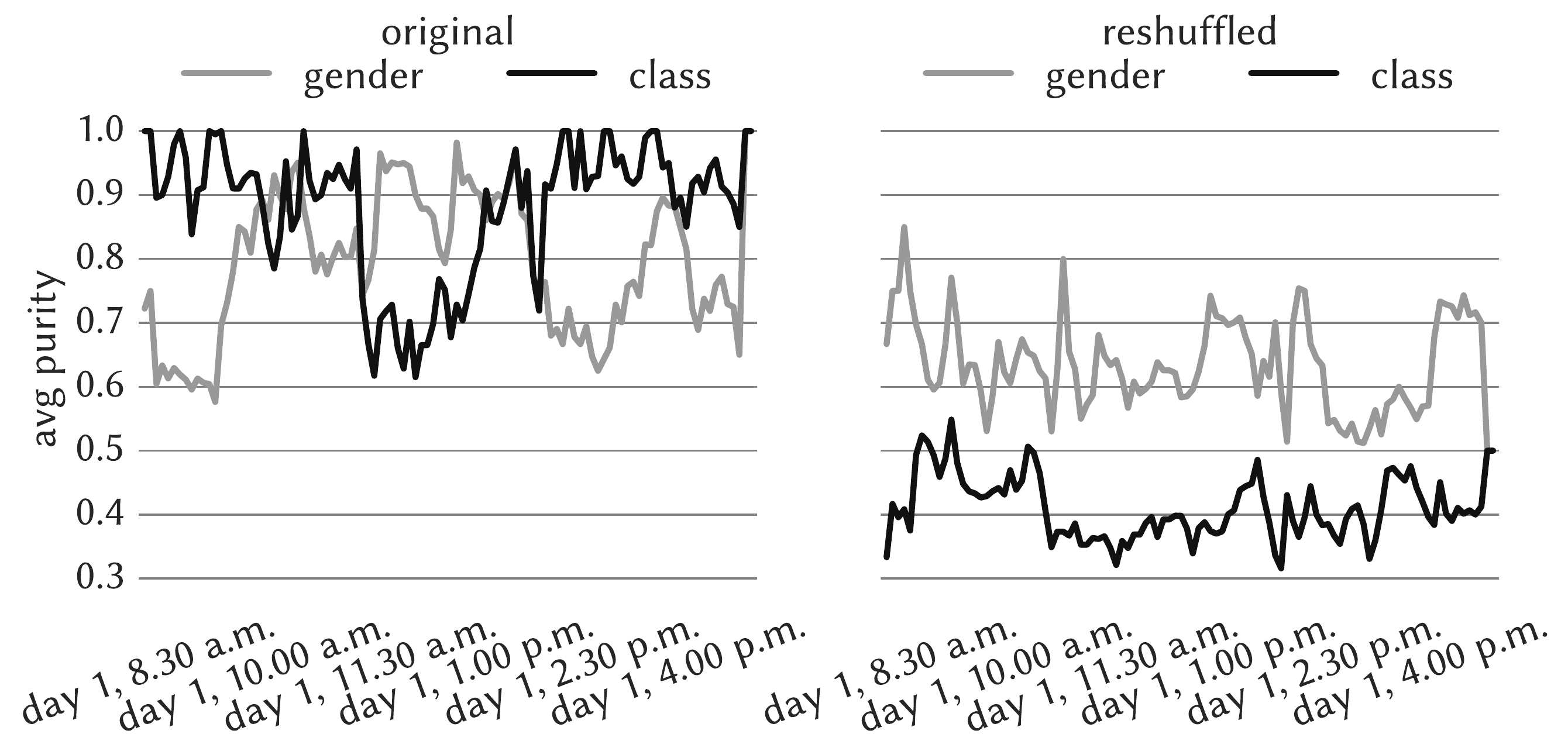}}
\vspace{-3mm}

\caption{\small \label{fig:purity} Temporal evolution (time on the $x$ axis) of average gender purity and average class purity ($y$ axis) of the maximal \spancores of the \textsf{PrimarySchool} dataset.
Original data on the left, reshuffled data on the right.}
\end{figure}

\spara{Interaction length.}
Finally, we analyze the duration of interactions of social groups in schools by studying the distribution of the size of the span of the maximal \spancores of the three datasets (Figure~\ref{fig:lengthdistribution}). All distributions are extremely skewed with broad tails: most maximal span-cores have duration less than $1$ hour, but
durations much larger than the average can also be observed. Interestingly, similar functional shapes are shown by the three datasets, confirming a robust statistical behavior.
We also note that similar robust broad distributions have been observed for simpler characteristics of human interactions such as the statistics of contact durations \cite{Stehle:2011,Fournet:PLOS2015}.
Outliers appear also at very large durations, especially for the \textsf{HongKong} dataset that has maximal \spancores lasting up to $83$ hours.
Group interactions of such long span are clearly abnormal and represent outliers in the distributions.
We will show, in the following of this section, how to exploit such outliers to detect both irregular contacts and anomalous temporal intervals.

\begin{figure}
\centerline{\includegraphics[width=1.05\columnwidth]{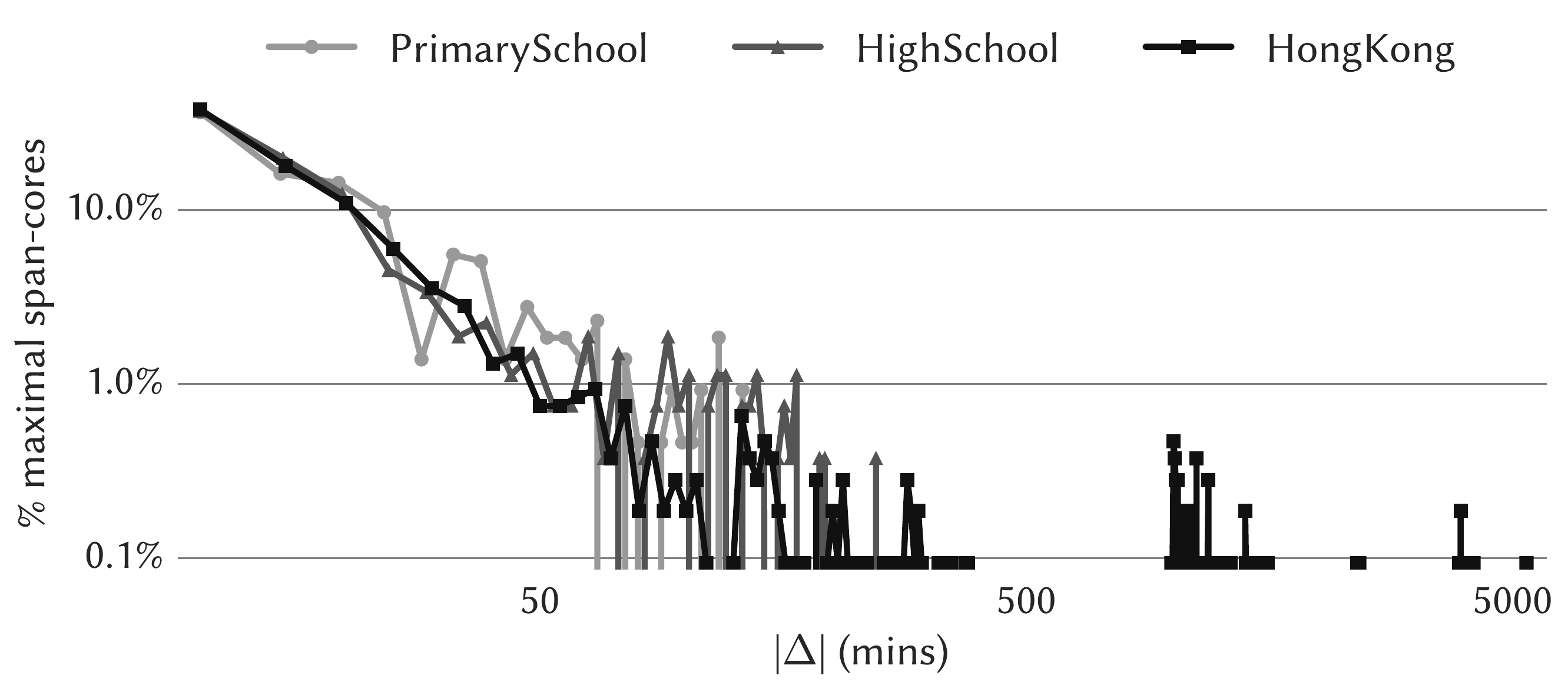}}
\vspace{-3mm}

\caption{\small \label{fig:lengthdistribution} Distribution of the size of the span $|\Delta|$ of the maximal \spancores. The $x$ axis reports the size of the span (in minutes), while the $y$ axis the percentage of maximal \spancores having a given size of the span.}
\end{figure}

\subsection{Anomaly detection}
The identification of anomalous behaviors in temporal networks has been the focus of several studies in the last few years~\cite{mongiovi2013netspot, sapienza2015detecting}.
Based on the above findings, we devise an extremely simple procedure to detect anomalous contacts and intervals of the \textsf{HongKong} dataset that exploits maximal \spancores.
The topmost plot of Figure~\ref{fig:anomaly} reports the number of contacts, i.e., edges, for each timestamp of the original \textsf{HongKong} dataset.
It is easy to notice that there is a lot of constant anomalous activity between school days and during the weekend, i.e., days six and seven.
Unexpectedly, the number of contacts per timestamp does not drop to zero because proximity sensors were left in each class, close to each other, at the end of the lessons. In order to automatically detect these steady activity patterns, we apply the following procedure: $(i)$ find a set of anomalously long temporal intervals supporting maximal \spancores, $(ii)$ identify anomalous vertices, and, $(iii)$ filter out anomalous contacts.

The first step of this procedure requires to find the set of temporal intervals $\mathcal{I} = \{\Delta \sqsubseteq T \mid C_{k,\Delta} \in \imcores \land |\Delta| > tr \}$ that are the span of a maximal \spancore $C_{k,\Delta}$ with size longer than a certain threshold $tr$.
Then, for each timestamp $t \in T$, select as anomalous all those vertices that appear in the \spancores $\{C_{1,\Delta} \mid \Delta \in \mathcal{I} \land t \in \Delta\}$, i.e., the \spancores of $k=1$ whose span is in $\mathcal{I}$ and contains $t$.
Finally, at each timestamp $t \in T$, filter out the contacts having at least an anomalous endpoint at time $t$.
Coherently to the distribution of the size of the span of the maximal \spancores, we select the threshold $tr = 22$ ($110$ minutes).
The results of this filtering procedure are shown in the middle plot of Figure~\ref{fig:anomaly}.
The number of contacts during school days remains substantially unchanged, while the activity noticeably decreases in-between.
Identifying as positives the contacts occurring when the school is closed and as negatives all the others (i.e., when the school is open), this approach achieves a precision of $0.91$ and a recall of $0.64$.

We can refine this anomaly detection process by identifying, in addition to anomalous contacts, also anomalous temporal intervals.
We define a timestamp $t \in T$ as anomalous if the ratio between the number of original contacts (top plot of Figure~\ref{fig:anomaly}) and the number of filtered contacts (middle plot of Figure~\ref{fig:anomaly}) exceeds a given threshold.
We apply this further filtering to the \textsf{HongKong} dataset with a threshold of $1.5$ and report the results in the bottommost plot of Figure~\ref{fig:anomaly}.
The number of contacts when the school is closed drops to zero, while the activity during school days is not modified, except for the last one, which is affected by the proximity to the end of the time domain.
The overall procedure yields a slightly higher value of precision, $0.93$, and substantially improves the recall to $0.99$.

\begin{figure}
\vspace{-1mm}
\centerline{\includegraphics[width=\columnwidth]{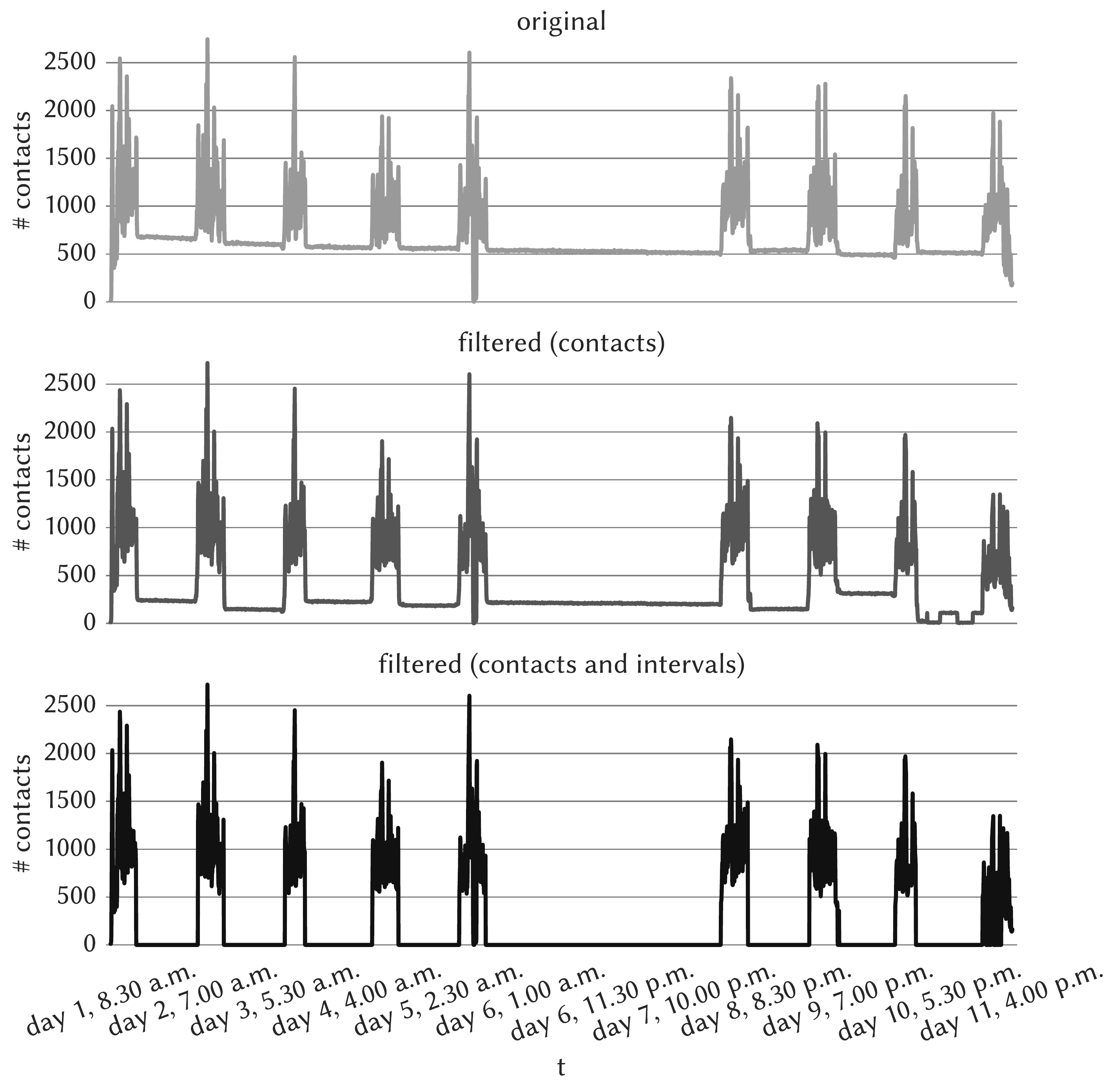}}
\vspace{-3mm}

\caption{\small \label{fig:anomaly} \textsf{HongKong} dataset: number of contacts ($y$ axis) per timestamp ($x$ axis) in the original data (top), after filtering anomalous contacts (middle), and after filtering anomalous contacts and intervals (bottom).}
\end{figure}

\section{Conclusions}
\label{sec:conclusions}


In this paper we introduced a notion of temporal core decomposition where each core is associated with its span, and developed efficient algorithms for computing all the span-cores, and only the maximal ones.
In our future work we will exploit span-cores for the computation of related notions, such as \emph{community search} or \emph{densest subgraph} in temporal networks. We will also study the role of maximal span-cores with large $\Delta$ in spreading processes on temporal networks. 
Furthermore,  span-cores represent features that can be used
for network finger-printing and classification, model validation, and could provide
support for new ways of visualizing large-scale time-varying graphs.

\end{document}